\documentclass[12pt,a4paper]{article}

\usepackage{cite}
\usepackage{amsmath}
\usepackage{amsthm}
\usepackage{amssymb}
\usepackage[utf8]{inputenc}
\usepackage{graphicx}
\usepackage{tikz}
\usepackage{hyperref}
\hypersetup{pdftex,colorlinks=true,allcolors=blue}
\usepackage{hypcap}
\usepackage{pgfplots}
\def\jgpapersize{}
\usepackage{subcaption}

\usetikzlibrary{decorations.markings}
\tikzset{->-/.style={decoration={
  markings,
  mark=at position #1 with {\arrow{>}}},postaction={decorate}}}
\tikzset{-<-/.style={decoration={
  markings,
  mark=at position #1 with {\arrow{<}}},postaction={decorate}}}

\newtheorem{theorem}{Theorem}
\newtheorem{lemma}[theorem]{Lemma}
\newtheorem{corrol}[theorem]{Corollary}

\def\defn#1{{\bf #1}}

\def\Real{{\mathbb R}}

\def\innerprod(#1,#2){{\left<#1\,,\,#2\right>}}
\def\Set#1{{\left\{#1\right\}}}
\def\qquadtext#1{\qquad\textup{#1}\qquad}
\def\qquadand{\qquadtext{and}}
\def\quadtext#1{\quad\textup{#1}\quad}
\def\quadand{\quadtext{and}}

\def\dfrac#1#2{\frac{d #1}{d #2}}

\usepackage{fancyhdr}
\fancypagestyle{firststyle}
{
   \fancyhf{}
   
\fancyfoot[C]{1}
\fancyfoot[R]{\footnotesize\bf File: \jobname.tex}
}
\def\TSchange#1{{\color{red} #1}}
\def\TSchange{}
\def\Mman{{\cal M}}

\def\Cdot{{\dot{C}}}
\def\Hdot{{\dot{H}}}

\def\Interval{{\cal I}}

\def\DixVec{N}
\def\DixonI/{\textbf{Dixon I}}
\def\DixonII/{\textbf{Dixon II}}
\def\DixonIII/{\textbf{Dixon III}}
\def\DixonNSS/{\textbf{Dixon1977}}

\def\zetaMultiDixon{\xi}
\def\TtwoTen{{{\cal \phi}}}

\def\deltafour{\delta^{(4)}}

\def\deltaFour{\delta^{(4)}}

\newcommand{\LF}{\left(}
\newcommand{\RF}{\right)}

\def\iSa{a}
\def\iSb{b}
\def\iSc{c}

\def\DixTube{{\cal U}_{\textup{DT}}}
\def\Vmu{{\underline\mu}}
\def\Vnu{{\underline\nu}}
\def\Vrho{{\underline\rho}}

\def\Vlambda{{\underline\lambda}}

\def\Nspace{{\DixVec^\perp}}
\def\Rootg{{\omega}}

\def\DEfullstop{\,.}
\def\DEcomma{\,,}
\def\DEnone{}

\def\ArbTen{{Y}}
\def\Jalt{Y}
\def\zetaalt{\theta}

\def\iMa{{{\mu}}}
\def\iMb{{{\nu}}}
\def\iMc{{{\rho}}}
\def\iMd{{{\kappa}}}

\def\n{{{\nabla}}}

\def\jgDchange{}

\def\jgnchange{}

\def\PiIField{\overline{\boldsymbol\Pi}}
\def\FixInd#1{{\{#1\}}}

\def\Proj{{\pi}}

\begin{document}
\jgpapersize

\title{The tensorial representation of the distributional stress–energy
quadrupole and its dynamics}

\author{Jonathan Gratus$^{1,2,3,*}$, 
Spyridon Talaganis$^{1,4}$}

\maketitle

\noindent
$^1$ Physics department, Lancaster University, Lancaster LA1 4YB, \TSchange{UK},
\\
$^2$ The Cockcroft Institute Daresbury Laboratory,
Daresbury, Warrington
WA4 4AD, UK.
\\
$^3$ \texttt{j.gratus@lancaster.ac.uk, orcid:0000-0003-1597-6084}
\\
$^4$ \texttt{s.talaganis@lancaster.ac.uk, orcid:0000-0003-0113-7546}
\\
$^*$ Corresponding author.
 

\begin{abstract}
We investigate stress-energy tensors constructed from the covariant
derivatives of delta functions on a worldline. Since covariant
derivatives are used all the components transform as tensors. We
derive the dynamical equations for the components, up to quadrupole
order.  The components do, however, depend in a non-tensorial way, on
a choice of a vector along the worldline.  We also derive a number of
important results about general multipoles, including that their
components are unique, and all multipoles can be written using
covariant derivatives. We show how the components of a multipole are
related to standard moments of a tensor field, by parallelly
transporting that tensor field.

\end{abstract}



\section{Introduction}

Finite size objects can be approximated using moments. This
makes sense when the object is viewed from a distance. The usual
definition of moments involves integrating over a spatial hypersurface
of an integrand which involves coordinates. For example one may define
the quadrupole of a rank (2,0) tensor $S^{\mu\nu}$ as 
$\int_{\text{space}} z^{a}z^{b} S^{\mu\nu} d^3z$, where $a,b=1,2,3$
and $\mu,\nu=0,\ldots,3$.
Such an object will not,
in general, be tensorial. Instead it will be highly dependent on the
choice of coordinates, and there will be no simple way of
transforming the expression from one coordinate system to
another. In addition it is necessary to \TSchange{choose} a foliation of
spacelike hypersurfaces over which one can perform the
integration. By contrast one can represent an extended object by a
distribution (in the Schwartz sense) over a worldline, which
may represent the ``centre'' of the object. This distribution is
tensorial in that it acts on appropriate test tensors to give a
number. Using this action one can find the transformation rules for
the moments. In this article such distributions will be called
\defn{multipoles}. The components of the multipoles are intimately
related to the standard moments, as we show below.

There are many extended objects one may wish to approximate. For
example a scalar field concentrated at one point in space can be
represented \TSchange{by a} scalar multipole, while the
current of an extended charge may be represented by a vector valued multipole.
In this article we are primary interested in modelling, as a multipole,
the stress-energy tensor of an extended massive object. Unlike the current
which can be used directly as a source for electromagnetism, the
distributional stress-energy tensor cannot be used directly in
Einstein's equations. This is because, unlike Maxwell's equations,
Einstein's equations are not linear. Instead one can use them as a
source for linearised gravity. 

Even if one is not concerned with the effect on gravity, considering
the stress-energy multipole is useful. This is because the constraint
that it must be \TSchange{divergenceless} tells us information about the dynamics
of the moments. For the monopole, there is just a single component,
which is constant and corresponds to the total mass.  For the dipole
there are 10 components whose dynamics are completely determined by
the Mathisson-Papapetrou-Tulczyjew-Dixon ODEs
\cite{mathisson1937neue,tulczyjew1959motion}.  By contrast for the
quadrupole, as well as the 40 components which have ODEs, there are 20
free components \TSchange{\cite{gratus2020distributional}}. One may consider the dynamics of
these free components to be constitutive relations. They could either
be posited as part of the model or derived by considering the
underlying matter which makes up the extended body. For example, we
would expect different results for dust than for a neutron star.

Multipole distributions can be represented in a number of ways.  In
\cite{gratus2020distributional} we discuss the advantages and
disadvantages of using two different expressions to represent a
multipole over a worldline. These are labelled the \defn{Ellis}
\cite{Ellis:1975rp} and \defn{Dixon}
\cite{dixon1964covariant,dixon1967description,dixon1970dynamics,DixonII,DixonIII}
representations. The Ellis representation uses partial derivatives and
can be applied to multipoles over any line on any manifold. However,
it requires complicated coordinate transformation rules for the
components, involving derivatives of the Jacobi matrix and integration
over the world line
\cite{gratus2020distributional,gratus2018correct}. The Dixon
representation uses covariant derivatives and has the significant
advantage that the components are tensors.  Another advantage is that
the multipole naturally splits into different orders of
``poles''. That is, it is a sum of a monopole, a dipole, a quadrupole
and so on. We refer to this in this as the \defn{Dixon split}.

The price one pays for this is that the manifold needs to have a
\TSchange{connection, and} one must choose a vector along the worldline, which we
call the \defn{Dixon vector}. In general changing the
Dixon vector will result in a complicated transformation, which not
only mixes orders but involves higher derivatives of the components.

With regards to multipoles over timelike \TSchange{worldlines} in
general relativity, these constraints are not a problem. Spacetime is
endowed with a connection, usually the Levi-Civita connection, and there
is a preferred vector over the worldline given by it's tangent. Other possible choices of the Dixon vector are \TSchange{discussed} in the conclusion.

The key result of this article, given in section \ref{ch_DivF}, is the
derivation of the dynamics of the components of the Dixon stress-energy
quadrupole. This is the generalisation of the
Mathisson-Papapetrou-Tulczyjew-Dixon equations. 
These are derived from the \TSchange{divergenceless} of the
stress-energy tensors and shows how the quadrupole couples to the curvature and
derivatives of the curvature. Similar equations have been derived by
Steinhoff and Puetzfeld
\cite{steinhoff2010multipolar}. However their method leads to an
implicit equation for the dynamics. By contrast our equations are
clearer with the time derivatives of the relevant components given
explicitly. 
The method, which involves commuting
covariant derivatives can be extended to arbitrary order
multipole. However\TSchange{,} as noted in \cite{gratus2020distributional}\TSchange{,} the equations do not completely
determine the dynamics of the quadrupole and must be augmented with
20 constitutive relations. 

In his work, Dixon makes two conjectures for the dynamics of the
components of a quadrupole, which we summarise in section
\ref{ch_Dix}. We can compare these to our dynamical equations and see
that neither of them couple to the curvature. Thus they do not  correspond to
the \TSchange{divergenceless} condition and are not the generalisation of
the Mathisson-Papapetrou–Tulczyjew–Dixon equations for the
quadrupole.

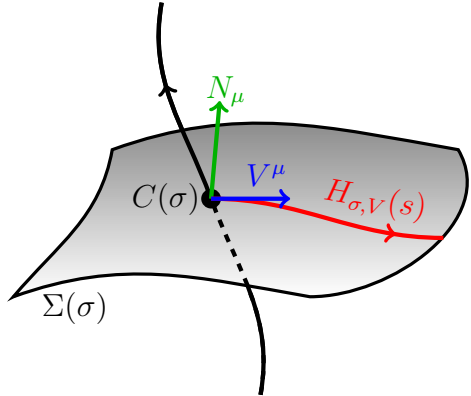
\begin{figure}
\centering
\begin{tikzpicture}[scale=1.3]
\draw [ultra thick,black,->-=0.8] (0.5,-2) to[out=80,in=-70] 
(0,0) to[out=110,in=-100] (-0.5,2) ;

\shadedraw [very thick] (-2,-1) to[out=20,in=170] 
   node [below,pos=0.2] {$\Sigma(\sigma)$} (1.,-1.) 
   to [out=0,in=-60] (2.5,.5) 
   to [out=170,in=20] (-1.,.5) 
   to [out=-100,in=50] cycle ;  

\draw [ultra thick,black,dashed] (0.5,-2) to[out=80,in=-70] (0,0) ;
\draw [ultra thick,black] (0,0) to[out=110,in=-100] (-0.5,2) ;

\fill (0,0) circle (0.1) node[left] {$C(\sigma)$}  ; 
\draw [ultra thick,black!30!green,->] (0,0) -- node
      [pos=1.1]{$\,\,\DixVec_{\iMa}$} (85:1) ;
\draw [ultra thick,red,->-=0.8] (0,0) to[out=0,in=-180] 
node[above,pos=0.7,rotate=-10] {$H_{\sigma,V}(s)$} (2.35,-.4) ;
\draw [ultra thick,blue,->] (0,0) -- node[above,pos=0.7] {$V^\mu$} (0:.8)  ;
\end{tikzpicture}
\caption{The worldline $C$ (black), with the Dixon geodesic
  hypersurface $\Sigma(\sigma)$ (grey) which intersects the worldline at
  $C(\sigma)$. Emanating from $C(\sigma)$ is the geodesic
  $H_{\sigma,V}(s)$ (red) which lies inside $\Sigma(\sigma)$. The
  tangent to $H_{\sigma,V}$ at $C(\sigma)$ is the vector $V^\mu$
  (blue), which is orthogonal to the Dixon vector $\DixVec_{\iMa}$
  (green).  All points in the Dixon geodesic hypersurface
  $\Sigma(\sigma)$ can be reached by a geodesic like $H_{\sigma,V}(s)$.
}
\label{fig_Dixon_Hyper}
\end{figure}

In addition to deriving the dynamical equations for stress-energy
quadrupole, in this article we also establish important results
about the Dixon representation of multipoles in general.  In
particular we show that all multipoles can be represented as Dixon
multipoles and that both the Dixon split and the components are
unique. We show how the components can be extracted from a
multipole by letting it act on appropriate test tensors.  All of these
results are valid for an arbitrary tensor multipole of arbitrary
order. Thus as well as the stress-energy multipole, they can also be
applied to electromagnetic current multipoles. 
The uniqueness of the splitting and the components is an
essential step in the derivation of the dynamical equations, as it is
applied to the divergence of the stress-energy tensor.

\vspace{1em}

This article is arranged as follows\TSchange{.} In section \ref{ch_SE} we give a
summary of the stress-energy tensor. In section \ref{ch_J} we give the
properties of arbitrary multipoles as described above, namely: the
Dixon split (section \ref{ch_Split}), the formula for the components
(section \ref{ch_ExtrComp}) and the demonstration that all multipoles
can be written as Dixon multipoles (section \ref{ch_Span}).  For this
we define a natural coordinate system, which is adapted to the
worldline and the Dixon vector (section \ref{ch_Hyp_R}). We also
define the Dixon geodesic hypersurface which is the hypersurface
generated from all geodesics which emanate from the worldline in a
direction orthogonal to the Dixon vector \TSchange{(figure
\ref{fig_Dixon_Hyper})}.  We make explicit the two point tensor
associated with parallel transport (section \ref{ch_Prll}) and give
the Taylor expansion of a test tensor which respects parallel
transport and the Dixon vector (section \ref{ch_Tay}).

As stated the definition of the spatial moments of a extended object
is in terms of the integral over space with respect to some coordinate
system.  Fortunately, once we have chosen a Dixon vector, one can use
Dixon \TSchange{geodesic} \TSchange{hypersurfaces, and the
  adapted coordinate system}. Alternatively \TSchange{one} can
parallel transport the tensor to the worldline. Thus the moments
defined this way are tensorial objects with \TSchange{tensorial}
transformation properties. Although they are still non-tensorially
dependent on the choice of the Dixon vector.  There is a natural way
of linking these moments with the components of the multipole
distribution. This is done by ``squeezing'' a regular
tensor. \TSchange{This is the process of} reducing a tensor's spatial
extent while keeping the quantity of matter constant. This is
demonstrated in section \ref{ch_Sqz}.

In section \ref{ch_DivF} we derive the dynamical equations for the
components of the stress-energy quadrupole. In section \ref{ch_Dix} we
compare these the equations proposed by Dixon.
Finally in section \ref{lm_phi_pll} we conclude.


\section{The stress-energy distribution}
\label{ch_SE}

We shall use the notation as defined in \cite{gratus2020distributional}. 
Let $\Mman$ be a spacetime with metric $g_{\iMa\iMb}$, signature
$(-,+,+,+)$, and the
Levi-Civita connection $\nabla_{\iMa}$ with
Christoffel symbol $\Gamma^{\iMa}_{\iMb\iMc}$. Here Greek indices 
$\iMa,\iMb=0,1,2,3$ and Latin indices $\iSa,\iSb=1,2,3$.  Let
$C:\Interval\to\Mman$ where $\Interval\subset\Real$ {is} the worldline
of the source\footnote{Even using proper time in Minkowski space, one
  cannot assume that $\Interval=\Real$ since it is possible to
  accelerate to lightlike infinity in finite proper time.}  with
components $C^{\iMa}(\sigma)$. At this point we do not assume that
$\sigma$ is proper time. Here we consider stress-energy tensors
$T^{\iMa\iMb}$ which are non-zero only on the worldline
$C^{\iMa}(\sigma)$, where it has Dirac--$\delta$ like
properties. Such stress-energy tensors are called
\defn{distributional}. 

{Since we are dealing with distributions it is most convenient to
  consider $T^{\iMa\iMb}$ as a tensor density\footnote{An integral
    over $\Mman$ must contain the measure $\Rootg$. There is therefore
    \jgDchange{the following} choice: one can choose $T^{\iMa\iMb}$ or
    $\TtwoTen_{\iMa\iMb}$ to be a density of weight 1, or put $\Rootg$
    explicitly in the integrand.  Here we have chosen to make
    $T^{\iMa\iMb}$ a density.} of weight 1. Thus $\Rootg^{-1}
  T^{\iMa\iMb}$ is a tensor, where
\begin{align}
\Rootg=\sqrt{-\det(g_{\iMa\iMb})}
\DEfullstop
\label{Intro_def_rootg}
\end{align}
  The definition of the covariant derivative of a tensor
  $\ArbTen^{\iMa\iMb\cdots}$
  density of weight 1 is given by
\begin{align}
\nabla_\iMa \ArbTen^{\iMb\iMc\cdots}
&=
\Rootg \nabla_\iMa (\Rootg^{-1} \ArbTen^{\iMb\iMc\cdots})
=
- 
\Gamma^\iMd_{\iMa\iMd}\, \ArbTen^{\iMb\iMc\cdots} 
+
\partial_\iMa \ArbTen^{\iMb\iMc\cdots} 
+
\Gamma^\iMb_{\iMa\iMd}\ArbTen^{\iMd\iMc\cdots} 
+
\Gamma^\iMc_{\iMa\iMd}\ArbTen^{\iMb\iMd\cdots} 
+\cdots
\label{Intro_ten_den}
\DEcomma
\end{align}
where $\Gamma^{\iMb}_{\iMa\iMc}$ are the Christoffel symbols.
In this article the term stress-energy tensor always refers to a
stress-energy tensor density of weight 1, even if not explicitly
stated. In addition the symbol $T^{\iMa\iMb}$ always refers to a
distributional stress-energy tensor density of weight 1 over the
worldline $C$.}

As already stated the Dixon representation depends crucially on a
choice of a vector field
$\DixVec_{\iMa}(\sigma)$ along the worldline $C$, called the
\defn{Dixon vector}. The only constraint on the choice of
$\DixVec_{\iMa}(\sigma)$ is that it is not orthogonal to the worldline
\jgnchange{$C$}, $\DixVec_{\iMa} \, \Cdot^{\iMa} \ne 0$. 
In section \ref{ch_DivF}, we need to 
project out the spatial components. Thus we scale $\DixVec_{\iMa}$ so that 
\begin{align}
\DixVec_{\iMa} \, \Cdot^{\iMa} = 1
\label{Intro_Tab_Dixon_N_NonOrth}
\DEfullstop
\end{align}
As long as the worldline $C$ is timelike, a natural choice of the
Dixon vector is $\Cdot^{\iMa}$, i.e. 
\TSchange{$\DixVec_{\iMa}
=
-g_{\iMa\iMb}\, \Cdot^{\iMb}
$}
but this is not the only choice.
{Having chosen $\DixVec_{\iMa}$, the \defn{Dixon} representation of a
multipole is given 
\cite[Equation (1.9)]{dixon1967description}\cite[Equation  (4.18),
  (7.4), (7.5)]{DixonII}} by
\begin{align}
T^{\iMa\iMb} = \sum_{r=0}^k {\frac{1}{r!}} \nabla_{\iMc_1} \cdots \nabla_{\iMc_r}
\int_\Interval \zetaMultiDixon^{\iMa\iMb \iMc_1\ldots\iMc_r}(\sigma)
\,
\deltaFour\big(x-C(\sigma)\big)\,d\sigma
\label{Intro_Tab_Dixon_Multi}
\DEfullstop
\end{align}
Tulczyjew \TSchange{\cite{tulczyjew1959motion}} calls this the
canonical \TSchange{form}, in the case when $\DixVec_{\iMa}=-\Cdot_{\iMa}$.

Since $T^{\iMa\iMb}$ is a stress-energy tensor, we have the symmetry
of the indices
\begin{align}
T^{\iMa\iMb} = T^{\iMb\iMa}
\label{SE_T_Symm}
\DEcomma
\end{align}
which leads to 
\begin{align}
\zetaMultiDixon^{\iMa\iMb \iMc_1\ldots\iMc_r}
=
\zetaMultiDixon^{\iMb\iMa \iMc_1\ldots\iMc_r}
\label{SE_Zeta_Symm}
\DEfullstop
\end{align}
We demand that the components 
$\zetaMultiDixon^{\iMa \iMb \iMc_1\ldots\iMc_k}$ are orthogonal to 
the vector $\DixVec_{\iMa}$,
\begin{align}
\DixVec_{\iMc_j}\
\zetaMultiDixon^{\iMa \iMb \iMc_1\ldots\iMc_k}
= 0
\label{Intro_Tab_Dixon_orthog}
\DEnone
\end{align}
for $j=1,\ldots,k$.
{The covariant derivatives do not commute.
Instead they give rise to curvature terms and lower the number of
derivatives. We therefore make the minimal choice and impose 
$\zetaMultiDixon^{\iMa \iMb \iMc_1\ldots\iMc_k}$ are symmetric in the
relevant indices.}
\begin{align}
\zetaMultiDixon^{\iMa \iMb \iMc_1\ldots\iMc_k} 
=
\zetaMultiDixon^{\iMa\iMb (\iMc_1\ldots\iMc_k)}
\label{Intro_Zeta_Dixon_sym}
\DEfullstop
\end{align}
{Since $T^{\iMa\iMb}$ is a tensor density this enables us to throw the
covariant derivative over onto the test tensor $\TtwoTen_{\iMa\iMb}$, giving 
\begin{align}
\int_\Mman T^{\iMa\iMb} \ \TtwoTen_{\iMa\iMb} \ d^4x 
&=
\sum_{r=0}^k (-1)^{r} {\frac{1}{r!}} \int_\Interval  
\zetaMultiDixon^{\iMa\iMb \iMc_1\ldots\iMc_r}(\sigma)\,
\big(\nabla_{\iMc_1} \cdots \nabla_{\iMc_r} \TtwoTen_{\iMa\iMb} \big)\big|_{C(\sigma)}
\ d\sigma
\label{Intro_Tab_Dixon_Multi_action}
\DEfullstop
\end{align}
This follow since if
$v^\iMa$ is a vector density of weight 1 then from (\ref{Intro_ten_den})
$\nabla_\iMa\,v^\iMa=\partial_\iMa\,v^\iMa$. 
}

At the quadrupole level the stress-energy distribution becomes
\begin{align}
T^{\mu\nu}
&=
\int_\Mman \xi^{\mu\nu}\,
\deltafour(z-C)\,d\sigma
+
\nabla_\rho\,
\int_\Mman \xi^{\mu\nu\rho}\,
\deltafour(z-C)\,d\sigma
+
\tfrac12 
\nabla_\rho\nabla_\sigma\,
\int_\Mman \xi^{\mu\nu\rho\sigma}\,
\deltafour(z-C)\,d\sigma
\label{Basic_Tmunu_Dix_quad}
\end{align}
where from (\ref{Intro_Tab_Dixon_orthog})
\begin{align}
\DixVec_\rho \xi^{\mu\nu\rho} = 0
\qquadand
\DixVec_\rho \xi^{\mu\nu\rho\sigma} = 0
\label{Basic_Dix_Cons_quaf}
\DEcomma
\end{align}
and from (\ref{SE_Zeta_Symm}) and (\ref{Intro_Zeta_Dixon_sym})
\begin{align}
\xi^{\mu\nu} = \xi^{\nu\mu}
,\quad
\xi^{\mu\nu\rho} = \xi^{\nu\mu\rho}
,\quad
\xi^{\mu\nu\rho\sigma} = \xi^{\nu\mu\rho\sigma}
\quadand
\xi^{\mu\nu\rho\sigma}
=
\xi^{\mu\nu\sigma\rho}
\label{SE_Xi_Symm_quad}
\DEfullstop
\end{align}

In this article we assume that $T^{\mu\nu}$ is \TSchange{divergenceless},
i.e.
\begin{align}
\nabla_\mu\,T^{\mu\nu} = 0
\label{SE_Div_Free}
\DEfullstop
\end{align}
This gives rise to dynamical equations for the components
$\xi^{\nu\mu}$, $\xi^{\nu\mu\rho}$ and $\xi^{\mu\nu\sigma\rho}$, which
we give below in theorem \ref{thm_DixQuad_Ten}.

\section{Properties of the Dixon representation of arbitrary
  distributions over worldlines}
\label{ch_J}

In this section, general details of multipoles are presented, which are
needed to analyse the stress-energy distribution. The most important
result we will use, is that the Dixon components are unique (section
\ref{ch_ExtrComp}).  This is needed so that when we take the
{divergence} of the stress-energy tensor and write that as a
Dixon distribution, we know that all the terms must vanish.

Since the results are true for all tensor distributions, not
simply the stress-energy tensor, we have chosen to derive the results
for an arbitrary tensor of rank $(m,0)$ and order $n$.

There are a number of concepts we need to define in order to show the
uniqueness of the components. First we need to establish the Dixon
geodesic hypersurfaces, the adapted coordinate system and the radial
vector (section \ref{ch_Hyp_R}), the notation for parallel transport
(section \ref{ch_Prll}) and the Dixon split (section \ref{ch_Split}).

The next step is to show that all multipoles are
Dixon multipoles (section \ref{ch_Span}). For this we need to be able
to take a Taylor expansion of the test tensor (section
\ref{ch_Tay}). Since we are dealing with tensors it is necessary to
transport the tensors around. There is no unique way of transporting
tensors and different choices will lead to different Taylor
expansions. The natural choice in this case is to use parallel
transport along the geodesics emanating from the worldline.

The final subsection of this section relates the moments of a regular
distribution with the components of a multipole. This is achieved by
squeezing the distribution.

\vspace{1em}

An arbitrary tensor density distribution
of rank $(m,0)$, weight 1, and order $n$
with support on $C$ is given by
\begin{align}
J^{\mu_1\cdots \mu_m} 
= 
\sum_{k=0}^{N} 
\frac{1}{k!}\, \nabla_{\rho_1}\ldots\nabla_{\rho_k}
\int_\Interval \zeta^{\mu_1\cdots \mu_m\rho_1\cdots\rho_k} 
\deltafour\big(x-C(\sigma)\big)\,d\sigma
\label{J_arb_J_form}
\DEcomma
\end{align}
where
\TSchange{
\begin{align}
\DixVec_{\rho_j} \zeta^{\mu_1\cdots\mu_m\rho_1\cdots \rho_k} = 0
\label{J_DixVec_Cons}
\end{align}
}
\TSchange{for $j=1,\dots,k$} and
\begin{align}
\zeta^{\mu_1\cdots\mu_m\rho_1\cdots \rho_k}
=
\zeta^{\mu_1\cdots\mu_m(\rho_1\cdots \rho_k)}
\label{J_Sym_J}
\DEfullstop
\end{align}
Unlike $\xi^{\mu\nu \rho_1\cdots \rho_k}$ we do not assume any symmetry
on the indices $\mu_1,\ldots,\mu_m$. 
For convenience when dealing with arbitrary tensors we replace the
indices $\mu_1\cdots\mu_m$ with the symbol $\Vmu$ so that
$J^{\mu_1\cdots\mu_m}=J^{\Vmu}$ and
$\zeta^{\mu_1\cdots\mu_m\rho_1\cdots \rho_k}
=\zeta^{\Vmu\rho_1\cdots\rho_k}$.

Let us introduce the notation for \TSchange{the symmetric sum} of 
multiple covariant derivatives
\begin{align}
\nabla^k_{\rho_1\cdots\rho_k} = \nabla_{(\rho_1}\ldots\nabla_{\rho_k)}
\label{J_def_nabla_k}
\DEfullstop
\end{align}

The result of applying a test tensor $\phi_{\Vmu}$ is given by
\begin{align}
\int_\Mman J^\Vmu\,\phi_{\Vmu}\, d^4 x
=
\sum_{k=0}^{n} 
\frac{1}{k!}\, (-1)^k 
\int_\Interval 
\zeta^{\Vmu\rho_1\cdots\rho_k} \,
(\nabla^k_{\rho_1\cdots\rho_k}\phi_{\Vmu})|_{C(\sigma)}
\,d\sigma
\label{J_J_phi}
\DEfullstop
\end{align}

Although the commutator of two covariant derivatives gives rise to
curvature terms, these are of a lower order so we can always write the
distribution using the symmetric sum of indices.

The key advantage of imposing the constraints
(\ref{J_DixVec_Cons}), (\ref{J_Sym_J}) is that they give rise
to unique components $\zeta^{\Vmu\rho_1\cdots\rho_k}$. We will see this
below is section \ref{ch_ExtrComp}.

\subsection{Notation and results for covariant derivatives}
\label{ch_Not}

Given a vector $V$ and a tensor $S^{\Vnu}_{\Vmu}$ introduce the
notation%
\footnote{We write vectors in the usual index notation as
  $V^\iMa$. However when a vector is an argument to a function or is a
  subscript we will drop the index and just write $V$.} 
\begin{align}
\nabla_V S^{\Vnu}_{\Vmu} = V^\rho \nabla_\rho S^{\Vnu}_{\Vmu}
\label{Not_def_nabla_V}
\end{align}
and the notation
\begin{align}
\nabla^r_V S^{\Vnu}_{\Vmu} = 
V^{\rho_1}\cdots V^{\rho_k}
\nabla^r_{\rho_1\cdots\rho_k} 
S^{\Vnu}_{\Vmu}
\label{Not_def_nabla_V_k}
\DEfullstop
\end{align}
Observe that in general 
$\nabla^r_V  S^{\Vnu}_{\Vmu} \ne (\nabla_V)^r
S^{\Vnu}_{\Vmu}$. However they do coincide when we have a geodesic.

\begin{lemma}
If $H(s)$ is a geodesic and $\phi_{\Vmu}$ is a tensor then
\begin{align}
\nabla^r_\Hdot \phi_{\Vmu}
=
(\nabla_\Hdot)^r\phi_{\Vmu}
\label{Not_nabla^r}
\DEfullstop
\end{align}

\end{lemma}
\begin{proof}
\begin{align*}
\nabla^r_\Hdot \phi_{\Vmu}
&=
\Hdot^{\rho_1}\cdots\Hdot^{\rho_r}\nabla^r_{\rho_1\cdots\rho_r} \phi_{\Vmu}
=
\Hdot^{\rho_1}\cdots\Hdot^{\rho_r}\nabla_{\rho_r}\cdots\nabla_{\rho_1} \phi_{\Vmu}
=
\Hdot^{\rho_1}\cdots\Hdot^{\rho_{r-1}}
\nabla_\Hdot\big(\nabla_{\rho_{r-1}}\cdots\nabla_{\rho_1} \phi_{\Vmu}\big)
\\&=
\nabla_\Hdot\big(\Hdot^{\rho_1}\cdots\Hdot^{\rho_{r-1}}
\nabla_{\rho_{r-1}}\cdots\nabla_{\rho_1} \phi_{\Vmu}\big)
\\&\qquad
-
\Big(
(\nabla_\Hdot\Hdot^{\rho_1})\Hdot^{\rho_2}\cdots\Hdot^{\rho_{r-1}}
+\cdots+
\Hdot^{\rho_1}\cdots\Hdot^{\rho_{r-2}}
(\nabla_\Hdot\Hdot^{\rho_{r-1}})
\Big)
\nabla_{\rho_{r-1}}\cdots\nabla_{\rho_1} \phi_{\Vmu}
\\&=
\nabla_\Hdot\big(\Hdot^{\rho_1}\cdots\Hdot^{\rho_{r-1}}
\nabla_{\rho_{r-1}}\cdots\nabla_{\rho_1} \phi_{\Vmu}\big)
\\&=\cdots=
(\nabla_\Hdot)^r\phi_{\Vmu}
\DEfullstop
\end{align*}
\end{proof}


\subsection{The Dixon geodesic hypersurface, 
Dixon \TSchange{adapted} coordinate system and the radial vector}
\label{ch_Hyp_R}

Given $\sigma\in\Interval$\TSchange{,} the set of \TSchange{vectors} which are perpendicular to
$\DixVec_\mu$ are denoted 
\begin{align}
\Nspace(\sigma) = \Set{
\textup{vectors $V^\mu$ at the point $C(\sigma)$ }
\big|\,N_\mu\,V^\mu=0}
\label{Basic_def_NVec}
\DEfullstop
\end{align}
Given $V^\mu\in\Nspace(\sigma)$ let 
$H_{\sigma,V}(s)$ be the geodesic satisfying
\begin{align}
H_{\sigma,V}(0)=C(\sigma)
\qquadand
\Hdot^\mu_{\sigma,V}(0)=V^\mu
\label{DivF_HorGeo_1}
\DEcomma
\end{align}
see figure \ref{fig_Dixon_Hyper}. Note that the parameter $s$ is not
normalised, so there is no constraint on the value of 
$g_{\mu\nu}\Hdot_{\sigma,V}^\mu(s)\Hdot_{\sigma,V}^\nu(s)
= g_{\mu\nu}V^\mu V^\nu$.
The domain of $H_{\sigma,V}(s)$ is distinct from
the domain $\Interval$ of $C(\sigma)$. It will always contain the initial value
$0$. Although it may not go all the way to
$\pm\infty$ it will go to the edge of the Dixon tube which is defined below.

It is useful to label the point $P(\sigma,V)\in\Mman$ reached from $C(\sigma)$
travelling along the geodesic $H_{\sigma,V}(s)$ a parameter
distance $1$, that is
\begin{align}
P(\sigma,V) = H_{\sigma,V}(1)
\label{Basic_def_P}
\DEfullstop
\end{align}
This gives the points for the geodesic $H_{\sigma,V}$ as
\begin{align}
H_{\sigma,V}(s) = H_{\sigma,s V}(1) = P(\sigma,s\,V)
\label{Basic_H_V_s}
\DEfullstop
\end{align}
All the points $P(\sigma,V)$, which are uniquely defined by
$(\sigma,V)$ form a neighbourhood of $C$.  We call this the
\defn{Dixon tube}, $\DixTube\subset \Mman$. Clearly $C\in\DixTube$.
There may be points where two different geodesics $H_{\sigma,V}$ and
$H_{\sigma',V'}$, with $\sigma\ne\sigma'$ intersect.  However these
will be at some distance from $C$ and outside of $\DixTube$.  In
addition there may be points which are unreachable from $C$.  The
set $\Set{(\sigma,V^\mu)\,\big|\,\sigma\in\Interval,\ V^\mu\in
  \Nspace(\sigma),\ P(\sigma,V)\in\DixTube}$, is diffeomorphic to
$\DixTube$.

Since we are dealing with
Schwartz distributions we demand that all test forms have compact support
which lie in $\DixTube$. This is not a significant restriction as any
other test function can be written, using partitions of unity, as the
sum of two test tensors, one with support inside $\DixTube$ and
another test tensor with support away from $C$. This second test
tensor, when acted upon by distributions on $C$ will always give zero.

To define the {Dixon adapted coordinate system} we require a frame
along $C$, $\Set{e_0,e_1,e_2,e_3}$ where $e^\mu(\sigma)$ is a vector
at the point $C(\sigma)$. We set $(e_0)^\mu=\Cdot^{\mu}$ and require
$(e_1)^\mu,(e_2)^\mu,(e_1)^\mu\in\Nspace(\sigma)$.  Thus for any
$\hat{U}^\mu\in\Nspace(\sigma)$ we can decompose it in terms of this
basis, giving 
$\hat{U}^\mu={U}^1 (e_1)^\mu + {U}^2 (e_2)^\mu + {U}^3 (e_3)^\mu$.

The \defn{Dixon adapted coordinate system} $(\sigma,z^1,z^2,z^3)$, is
given on the Dixon tube such that
\begin{align}
\sigma|_{P(\sigma',V)}=P^0(\sigma',V) = \sigma'
\qquadand
z^\iSa|_{P(\sigma',V)}=P^a(\sigma',V) = V^\iSa
\label{DivF_DixCoord}
\DEfullstop
\end{align}
We set $z^0=\sigma$ so that we can label the coordinates of a point
$p$ by $(p^0,p^1,p^2,p^3)$. We use Latin indices $a,b,\ldots=1,2,3$ and use
the summation convention over Latin indices to sum from 1 to 3. Any
reference in the article to adapted coordinates, or whenever Latin
indices are used we always mean the Dixon adapted coordinate system.

\begin{lemma}
In the Dixon \TSchange{adapted} coordinate system
\begin{align}
H_{\sigma,V}^0(s) = \sigma, \quad
H_{\sigma,V}^a(s) = s\,V^a, \quad
\Hdot_{\sigma,V}^0 = 0
\quadand
\Hdot_{\sigma,V}^a(s) = V^a
\label{Basic_DixCords_H}
\DEfullstop
\end{align}
\end{lemma}
\begin{proof}
Equation (\ref{Basic_DixCords_H}.1)
follows\footnote{(\ref{Basic_DixCords_H}.1) refers to the first
  equation in (\ref{Basic_DixCords_H}).} from 
$H_{\sigma',V}^0(s) = H_{\sigma',sV}^0(1) = P(\sigma',s V)^0=\sigma'$. 
(\ref{Basic_DixCords_H}.2) follows from 
$H_{\sigma,V}^a(s) = H_{\sigma,sV}^a(1) = P(\sigma,s V)^a=s V^a$. 
For (\ref{Basic_DixCords_H}.3) and  (\ref{Basic_DixCords_H}.4) 
\begin{align*}
\Hdot_{\sigma,V}^0 
=
\dfrac{}{s} (\sigma)
=
0
\qquadand
\Hdot_{\sigma,V}^a 
=
\dfrac{}{s} ({H_{\sigma,sV}^a(s)})
=
\dfrac{}{s} (s\, V^a)
=
V^a
\DEfullstop
\end{align*}
\end{proof}

\vspace{1em}

The \defn{radial vector} field, $R^\mu$, is a vector field on $\Mman$ given by 
\begin{align}
R^\mu|_{P(\sigma,V)} = \Hdot^\mu_{\sigma,V}(1)
\label{DivF_R_1}
\DEfullstop
\end{align}
It is a key ingredient for the Dixon split. In \cite{gratus2020distributional}, the radial
vector was not defined completely. Instead, only some of the necessary
properties of the radial vector field were given, in order to give
the Dixon split up to quadrupole order. Here the specific radial
vector field is defined in order to give the Dixon split to arbitrary
order.

\begin{lemma}
The radial vector has the properties that 
\begin{align}
R^\mu|_C = 0,\qquad
\nabla_{U_1} R^\mu = U_1,\qquad
{U_1^{\rho_1} U_2^{\rho_2}\cdots U_r^{\rho_r}}\nabla^r_{\rho_1\cdots\rho_r} R^\mu = 0
\label{DivF_R_prop}
\DEcomma
\end{align}
for $r\ge 2$, $U^{\rho_i}_i\in \Nspace(\sigma)$.

In the Dixon \TSchange{adapted} coordinate system
\begin{align}
R^0 = 0 \qquadand
R^a = z^\iSa
\label{DivF_R_DixCoord}
\DEfullstop
\end{align}
\end{lemma}

\begin{proof}
Clearly $R^\mu|_C=0$.

Fix $\sigma$ and $V$ and let $H=H_{\sigma,V}$ so that 
$H(s)=P(\sigma,s V)$ and $s \Hdot^\mu(s)=R^\mu|_{H(s)}$.
Then
\begin{align*}
\nabla_{\Hdot(s)} R^\mu 
=
\nabla_{\Hdot(s)} \big(s\Hdot^\mu(s)\big)
=
\Hdot^\mu(s)
\DEfullstop
\end{align*}
Hence setting $V=U_1$ and $s=0$ we get (\ref{DivF_R_prop}.2).

For $r\ge 2$ we have
from (\ref{Not_nabla^r}) 
\begin{align*}
\nabla^r_{\Hdot} R^\mu
&=
(\nabla_{\Hdot})^r R^\mu
=
(\nabla_{\Hdot})^{r-1} (\nabla_{\Hdot} R^\mu)
=
(\nabla_{\Hdot})^{r-1} \Hdot^\mu
=
0
\DEfullstop
\end{align*}
Hence setting $s=0$ we have $\nabla^r_{V} R^\mu=0$. Now setting
$V$ equal to combinations of sums
$V=U_1\pm U_2 \pm \cdots \pm U_r$, and then taking the sums, we are
left with (\ref{DivF_R_prop}.3).

From (\ref{Basic_DixCords_H}.3) we have (\ref{DivF_R_DixCoord}.1).
From (\ref{Basic_DixCords_H}.4) and (\ref{DivF_DixCoord}.2)
$R^a|_{P(\sigma,V)}=\Hdot_{\sigma,V}^a=V^a=z^a|_{P(\sigma,V)}$ giving
(\ref{DivF_R_DixCoord}.2)
\end{proof}

We use the Dixon vector to define the \defn{Dixon geodesic hypersurface}
$\Sigma(\sigma)\subset\DixTube$ as the set
\begin{align}
\Sigma(\sigma)=\Set{P(\sigma,V)\in\DixTube\,\big|\,V^\mu\in\Nspace(\sigma)}
\label{Basic_def_Dix_Hyp}
\DEfullstop
\end{align}
This is given in figure \ref{fig_Dixon_Hyper}.

\subsection{Parallel transport}
\label{ch_Prll}

The Dixon geodesic hypersurfaces (\ref{Basic_def_Dix_Hyp}) are
constructed from geodesics emanating from the worldline, with initial
tangent orthogonal to the Dixon vector. It is necessary to parallel
transport tensors along these geodesics. We therefore define a two
point tensor $\Pi^\mu_\nu|^p_q$ where $p$ and $q$ lie along the same
geodesic. This tensor can be used to parallel transport tensors
along the geodesics and hence defined on the Dixon geodesics
hypersurfaces. This is seen in lemmas \ref{lm_Prll_Trans_Vec},
\ref{lm_phi_pll} and \ref{lm_phi_Vmu_pll} below.

Given $\sigma\in\Interval$ and $V^\mu\in\Nspace(\sigma)$ and given
$s_0,s_1\in\Real$ such that $P(\sigma,s_0 V)\in\DixTube$ and 
$P(\sigma,s_1 V)\in\DixTube$,
let 
\begin{align}
\Pi^\mu_\nu\big|^{P(\sigma,s_0 V)}_{P(\sigma,s_1 V)}
:\Set{\textup{vectors at $P(\sigma,s_0 V)$}}
\longrightarrow 
\Set{\textup{vectors at $P(\sigma,s_1 V)$}}
\label{Prll_Pi_Map}
\DEnone
\end{align}
be the two point
tensor satisfying the differential equation
\begin{align}
\dfrac{}{s} \Pi^\mu_\nu\big|^{P(\sigma,s_0 V)}_{P(\sigma,s V)}
=
-
\Gamma^\mu_{\TSchange{\lambda} \rho}|_{P(\sigma,s V)}\, 
\Hdot_{\sigma,V}^\rho(s) \ \Pi^{\TSchange{\lambda}}_\nu\big|^{P(\sigma,s_0 V)}_{P(\sigma,s V)}
\qquadand
\Pi^\mu_\nu\big|^{P(\sigma,s_0 V)}_{P(\sigma,s_0 V)} = \delta^\mu_\nu 
\label{Prll_def_Pi}
\DEfullstop
\end{align}
Observe that 
\begin{align}
\Pi^\mu_\nu\big|^{P(\sigma,s_0 V)}_{P(\sigma,s_1 V)}\,
\Pi^\nu_\rho|^{P(\sigma,s_1 V)}_{P(\sigma,s_2 V)}
&=
\Pi^\mu_\rho|^{P(\sigma,s_0 V)}_{P(\sigma,s_2 V)}\,
\qquadand
\Pi^\mu_\nu\big|^{P(\sigma,s_0 V)}_{P(\sigma,s_1 V)}\,
\Pi^\nu_\rho|^{P(\sigma,s_1 V)}_{P(\sigma,s_0 V)}
=
\delta^\mu_\rho
\label{Prll_Pi_Group}
\DEnone
\end{align}
so that from (\ref{Prll_Pi_Group}.2)
\begin{align}
\Pi^\mu_\nu\big|^{P(\sigma,s_1 V)}_{P(\sigma,s_1 V)}
=
\delta^\mu_\nu
\label{Prll_Pi_Iden}
\DEfullstop
\end{align}

From the definition (\ref{Prll_def_Pi}) it may look like the parallel
transport is dependent on the choice of parameterisation of the
geodesic $H_{\sigma,V}$. We see here that this is not the
case.
\begin{lemma}
If $s_0 V^\mu=\hat{s}_0 \hat{V}^\mu$ and 
$s_1 V^\mu=\hat{s}_1 \hat{V}^\mu$ then
\begin{align}
\Pi^\mu_\nu\big|^{P(\sigma,s_0 V)}_{P(\sigma,s_1 V)}
=
\Pi^\mu_\nu\big|^{P(\sigma,\hat{s}_0 \hat{V})}_{P(\sigma,\hat{s}_1 \hat{V})}
\label{Prll_Pi_Ind}
\DEfullstop
\end{align}
I.e. given $p,q\in\DixTube$ along the same $H_{\sigma,V}$, then 
$\Pi^\mu_\nu\big|^{p}_{q}$ is
independent of the choice of parameterisation of $H_{\sigma,V}$.
\end{lemma}
\begin{proof}
Let $V^\mu=\kappa \hat{V}^\mu$ so that
$s_0=\hat{s}_0/\kappa$ and $s_1=\hat{s}_1/\kappa$. Setting 
$\dfrac{}{s}=\kappa \,\dfrac{}{\hat{s}}$ then
(\ref{Prll_def_Pi}.1) becomes
\begin{align*}
\kappa \dfrac{}{\hat{s}} \Pi^\mu_\nu\big|^{P(\sigma,s_0 V)}_{P(\sigma,\hat{s} V)}
=
-
\Gamma^\mu_{\TSchange{\lambda} \rho}|_{P(\sigma,\hat{s} V)}\, 
\kappa \dfrac{}{\hat{s}}H_{\sigma,V}^\rho(\hat{s}) \ 
\Pi^{\TSchange{\lambda}}_\nu\big|^{P(\sigma,s_0 V)}_{P(\sigma,\hat{s} V)}
\DEcomma
\end{align*}
and the $\kappa$'s cancel.
Setting $V^\mu\to \kappa \hat{V}^\mu$ and likewise for $s_0$ and $s_1$ we see
the ODE (\ref{Prll_def_Pi}.1) is independent of scaling of
$s,V,s_0,s_1$. The initial condition follows from (\ref{Prll_Pi_Iden}). 
\end{proof}


\begin{lemma}
\label{lm_Prll_Trans_Vec}
Given $(\sigma,V)$, a vector $\hat{U}^\mu$ the point 
$P(\sigma,s_0 V)$, and vector field $U^\mu(s)$ along $H_{\sigma,V}$ then
\begin{align}
\nabla_{\Hdot_{\sigma,V}} U^\mu = 0
\qquadand
U^\mu|_{P(\sigma,s_0 V)} = \hat{U}^\mu
\label{Prll_nabla_U}
\DEcomma
\end{align}
{if and only if}
\begin{align}
U^\mu|_{P(\sigma,s V)} = 
\Pi^\mu_\nu\big|^{P(\sigma,s_0 V)}_{P(\sigma,s V)} \ \hat{U}^\nu
\label{Prll_Tranp_U}
\DEfullstop
\end{align}
\end{lemma}
\begin{proof}
Let $H(s)=H_{\sigma,V}(s)$ then
\begin{align*}
\nabla_{\Hdot(s_1)} \big(\Pi^\mu_\nu\big|^{P(\sigma,s_0 V)}_{P(\sigma,s V)} \ \hat{U}^\nu\big)
&=
\dfrac{}{s} 
\big(\Pi^\mu_\nu\big|^{P(\sigma,s_0 V)}_{P(\sigma,s V)} 
\  
\hat{U}^\nu \big)\Big|_{s=s_1}
+
\Gamma^\mu_{\nu\rho}|_{H(s_1)}\ \Hdot^\rho(s_1)\ U^\nu(s_1)
\\&=
\dfrac{}{s} 
\big(\Pi^\mu_\nu\big|^{P(\sigma,s_0 V)}_{P(\sigma,s V)} \big)\Big|_{s=s_1}
\  
\hat{U}^\nu 
+
\Gamma^\mu_{\nu\rho}|_{H(s_1)}\ \Hdot^\rho(s_1)\ U^\nu(s_1)
\\&=
-
\Gamma^\mu_{\sigma\rho}|_{H(s)}\,
\Hdot^\rho(s)\ \Pi^\sigma_\nu\big|^{P(\sigma,s_0 V)}_{P(\sigma,s V)}\  \Big|_{s=s_1}
\hat{U}^\nu 
+
\Gamma^\mu_{\nu\rho}|_{H(s_1)}\ \Hdot^\rho(s_1)\ U^\nu(s_1)
\\&=
-
\Gamma^\mu_{\sigma\rho}|_{H(s_1)}\, \Hdot^\rho(s_1)\ \Pi^\sigma_\nu\big|^{P(\sigma,s_0 V)}_{P(\sigma,s_1 V)}\ 
\hat{U}^\nu 
+
\Gamma^\mu_{\nu\rho}|_{H(s_1)}\ \Hdot^\rho(s_1)\ U^\nu(s_1)
\\&=
0
\DEfullstop
\end{align*}
Hence result.
\end{proof}

\begin{lemma}
\label{lm_phi_pll}
Given a 1--form $\hat{\phi}_\mu$ at the point $P(\sigma,s_0 V)$, then
\begin{align}
\nabla_{\Hdot_{\sigma,V}} \phi_\mu = 0
\quadand
\phi_\mu|_{P(\sigma,s_0 V)}=\hat{\phi}_\mu
\label{Prll_Tranp_phi_nabla}
\DEcomma
\end{align}
if and only if
\begin{align}
\phi_\mu|_{P(\sigma,s V)} = \Pi^\nu_\mu\big|^{P(\sigma,s V)}_{P(\sigma,s_0 V)} \ \hat{\phi}_\nu
\label{Prll_Tranp_phi}
\DEfullstop
\end{align}
\end{lemma}
\begin{proof}
Given a vector $\hat{U}\in T_{H(s_0)}\Mman$ and a vector field $U(s)$
along $H$ then such that 
(\ref{Prll_Tranp_U}) then
\begin{align*}
\nabla_{\Hdot(s_1)} 
\big( \Pi^\nu_\mu\big|^{P(\sigma,s V)}_{P(\sigma,s_0 V)} \ \hat{\phi}_\nu\big)
\big|_{s_1}
U^\mu(s_1)
&=
\nabla_{\Hdot(s_1)} 
\big( \Pi^\nu_\mu\big|^{P(\sigma,s V)}_{P(\sigma,s_0 V)} \ \hat{\phi}_\nu\big)
\big|_{s_1}
U^\mu(s_1)
+
\Pi^\nu_\mu(s_1,s_0)\ \hat{\phi}_\nu
\nabla_{\Hdot(s_1)} U^\mu
\\&=
\nabla_{\Hdot(s_1)} 
\big( \Pi^\nu_\mu\big|^{P(\sigma,s V)}_{P(\sigma,s_0 V)} \ \hat{\phi}_\nu U^\mu(s) \big)
\big|_{s_1}
\\&=
\dfrac{}{s}
\big( \Pi^\nu_\mu\big|^{P(\sigma,s V)}_{P(\sigma,s_0 V)}
\ \hat{\phi}_\nu 
\ \Pi^\mu_\rho\big|^{P(\sigma,s_0 V)}_{P(\sigma,s V)} 
\ \hat{U}^\rho \big)
\big|_{s_1}
=
\dfrac{}{s}
\big( \hat{\phi}_\nu \hat{U}^\nu \big)
\big|_{s_1}
=
0
\DEfullstop
\end{align*}
Hence (\ref{Prll_Tranp_phi}) if and only if (\ref{Prll_Tranp_phi_nabla}).
\end{proof}

We need to extend the results of lemma \ref{lm_phi_pll} for tensors with
arbitrary number of indices, i.e. $\phi_\Vmu$. Let
\begin{align}
\Pi^\Vmu_\Vnu\big|^{P(\sigma,s_0 V)}_{P(\sigma,s V)}
=
\Pi^{\mu_1}_{\nu_1}\big|^{P(\sigma,s_0 V)}_{P(\sigma,s V)}
\cdots
\Pi^{\mu_m}_{\nu_m}\big|^{P(\sigma,s_0 V)}_{P(\sigma,s V)}
\label{Prll_dep_Pi_Vmu}
\DEfullstop
\end{align}
\begin{lemma}
\label{lm_phi_Vmu_pll}
Given a tensor $\hat{\phi}_\Vmu$ at the point  at the point $P(\sigma,s_0 V)$, then
\begin{align}
\nabla_{\Hdot_{\sigma,V}} \phi_\Vmu = 0
\quadand
\phi_\Vmu|_{P(\sigma,s_0 V)}=\hat{\phi}_\Vmu
\label{Prll_Tranp_phiVmu_nabla}
\DEcomma
\end{align}
if and only if
\begin{align}
\phi_\Vmu|_{P(\sigma,s V)} = \Pi^\Vnu_\Vmu\big|^{P(\sigma,s V)}_{P(\sigma,s_0 V)} \ \hat{\phi}_\Vnu
\label{Prll_Tranp_phiVmu}
\DEfullstop
\end{align}
\end{lemma}
\begin{proof}
Set $\phi_\Vmu$ to be the outer product of 1--forms
$\phi_\Vmu=\phi^1_{\mu_1}\cdots\phi^m_{\mu_m}$. Then apply lemma \ref{lm_phi_pll}.
\end{proof}

\begin{lemma}
Let $\phi_\Vmu$ be a 1--form field over $H_{\sigma,V}$. 
Then
\begin{align}
\frac{d^k}{ds^k} 
\Big(\Pi^\Vmu_\nu\big|^{P(\sigma,s V)}_{P(\sigma,s_0 V)}  \phi_\Vmu|_{P(\sigma,s V)}\Big)
=
\nabla_{\Hdot_{\sigma,V}}^k \phi_\nu|_{P(\sigma,s_0 V)}
\label{Prll_higher_deriv}
\DEfullstop
\end{align}
\end{lemma}
\begin{proof}
Proof by induction on $r$. Clearly true when $r=0$ and $r=1$, Let
$H(s)=H_{\sigma,V}(s)$. Assume
true for $r$,
\begin{align*}
\frac{d^{r+1}}{ds^{r+1}} 
\Big(\Pi^\Vmu_\Vnu\big|^{P(\sigma,s_0 V)}_{P(\sigma,s V)}  \phi_\Vmu|_{P(\sigma,s V)}\Big)
&=
\dfrac{}{s_2} \bigg(
\frac{d^{r}}{ds_1^{r}} \Big(\Pi^\Vmu_\Vnu\big|^{P(\sigma,s_0 V)}_{P(\sigma,s_1 V)} 
\phi_\Vmu|_{P(\sigma,s_1 V)}\Big)\Big|_{s_1=s_2} 
\bigg)\bigg|_{s_2=s_0}
\\&=
\dfrac{}{s_2} \bigg(
\frac{d^{r}}{ds_1^{r}} \Big(
\Pi^\Vmu_\Vrho\big|^{P(\sigma,s_0 V)}_{P(\sigma,s_2 V)} \ 
\Pi^\Vrho_\Vnu\big|^{P(\sigma,s_2 V)}_{P(\sigma,s_1 V)}\ 
\phi_\Vmu|_{P(\sigma,s_1 V)}\Big)\Big|_{s_1=s_2} 
\bigg)\bigg|_{s_2=s_0}
\\&=
\dfrac{}{s_2} \bigg(
\Pi^\Vmu_\Vrho\big|^{P(\sigma,s_0 V)}_{P(\sigma,s_2 V)} \ 
\frac{d^{r}}{ds_1^{r}} \Big(
\Pi^\Vrho_\Vnu\big|^{P(\sigma,s_2 V)}_{P(\sigma,s_1 V)}\ 
\phi_\Vmu|_{P(\sigma,s_1 V)}\Big)\Big|_{s_1=s_2} 
\bigg)\bigg|_{s_2=s_0}
\\&=
\dfrac{}{s_2} \bigg(
\Pi^\Vmu_\Vrho\big|^{P(\sigma,s_0 V)}_{P(\sigma,s_2 V)} \ 
\nabla_{\Hdot}^r \phi_\Vmu  |_{P(\sigma,s_2 V)}
\bigg)\bigg|_{s_2=s_0}
\\&=
\nabla_\Hdot \nabla^r_{\Hdot} \phi_\Vnu |_{P(\sigma,s_0 V)}
=
\nabla_\Hdot (\nabla_{\Hdot})^r \phi_\Vnu |_{P(\sigma,s_0 V)}
=
\nabla^{r+1}_\Hdot\phi_\Vnu |_{P(\sigma,s_0 V)}
\DEfullstop
\end{align*}

\end{proof}


It is necessary to define a collection of tensor fields
$\PiIField^{\FixInd\Vnu}_\Vmu$, which are used to define the test tensors
needed to extract the components $\zeta^{\Vmu\rho_1\cdots\rho_k}$ in
section \ref{ch_ExtrComp}.  These tensors are parallel and form a
basis for tensors. These will be necessary to define a Taylor
expansion of the test tensors and for extracting the components of a
multipole.
\begin{align}
\PiIField^{\FixInd\Vnu}_\Vmu|_{P(\sigma,V)} = \Pi^\Vnu_\Vmu
\big|^{P(\sigma,V)}_{C(\sigma)}
\label{Prll_def_Pi_field}
\DEfullstop
\end{align}
The reason for placing the $\Vnu$ in curly brackets is because these
are tensor indices referring to the point $C(\sigma)$. Therefore the
covariant derivative $\nabla_\rho$ does not produce the Christoffel symbols
$\Gamma^{\nu_i}_{\rho\sigma}$. I.e.
\begin{align}
\nabla_\rho \PiIField^{\FixInd\Vnu}_\Vmu
&=
\partial_\rho \PiIField^{\FixInd\Vnu}_\Vmu
-
\Gamma^{\sigma}_{\rho\mu_1}\,
\PiIField^{\FixInd\Vnu}_{\sigma\mu_2\cdots\mu_m}
- \cdots - 
\Gamma^{\sigma}_{\rho\mu_m}\,
\PiIField^{\FixInd\Vnu}_{\mu_1\cdots\mu_{m-1}\sigma}
\label{Prll_nabla_PiI}
\DEfullstop
\end{align}
Hence from lemma \ref{lm_phi_Vmu_pll},
\begin{align}
\nabla_{\Hdot_{\sigma,V}} \PiIField^{\FixInd\Vnu}_\Vmu = 0
\label{Prll_nabla_as_PiI}
\DEfullstop
\end{align}

\begin{lemma}
In adapted coordinates
\begin{align}
\nabla^k_{a_1\cdots a_k} \PiIField^{\FixInd\Vnu}_\Vmu = 0
\label{Prll_multideriv_Pi}
\DEfullstop
\end{align}
\end{lemma}
\begin{proof}
Fix $\sigma$ and $V^\mu\in\Nspace$ and let $H=H_{\sigma,V}$ so that 
$H(s)=P(\sigma,s V)$. Since $H$ is a geodesic
\begin{align*}
V^{a_1}\cdots V^{a_k} \nabla^k_{a_1\cdots a_k} 
\PiIField^{\FixInd\Vnu}_\Vmu \big|_{C(\sigma)}
=
\nabla_{V}^k \PiIField^{\FixInd\Vnu}_\Vmu \big|_{C(\sigma)}
=
\nabla_{\Hdot}^k \PiIField^{\FixInd\Vnu}_\Vmu \big|_{C(\sigma)}
=
0
\DEfullstop
\end{align*}
Now setting
$V^\mu$ equal to combinations of sums
$V=U^\mu_1\pm U^\mu_2 \pm \cdots \pm U^\mu_r$, and then taking the sums, we are
left with 
\begin{align*}
U_1^{a_1}\cdots U_k^{a_k} \nabla^k_{a_1\cdots a_k} 
\PiIField^{\FixInd\Vnu}_\Vmu \big|_{C(\sigma)}
= 0
\DEcomma
\end{align*}
since $\nabla^k_{a_1\cdots a_k}$ is symmetric in $a_1,\ldots,a_k$. 
Since this is true for all vectors $V^\mu=U^\mu_1,\dots,U^\mu_r$ at $C(\sigma)$
we have
(\ref{Prll_multideriv_Pi})

\end{proof}


\subsection{The Dixon split}
\label{ch_Split}

Given an arbitrary Dixon multipole $J^\Vmu$, or order $n$,
let 
\begin{align}
J^{\Vmu} = \sum_{k=0}^{n} J^{\Vmu}_{(k)}
\label{Split_arb_J}
\DEcomma
\end{align}
where
\begin{align}
J^{\Vmu}_{(k)} 
= 
\frac{1}{k!}\, \nabla_{\rho_1\ldots\rho_k}
\int \zeta^{\Vmu,\rho_1\cdots\rho_k} 
\deltafour(z-C)\,d\sigma
\label{Split_arb_J_form}
\DEfullstop
\end{align}
I.e.
\begin{align}
J^{\Vmu}_{(k)}[\phi_{\Vmu}]
= 
(-1)^k
\frac{1}{k!}\, 
\int \zeta^{\Vmu,\rho_1\cdots\rho_k} 
\nabla^k_{\rho_1\ldots\rho_k}
\phi_{\Vmu}
\,d\sigma
\label{Split_arb_J_form_act}
\DEfullstop
\end{align}
We call the set $\Set{J^{\Vmu}_{(k)}}$ the \defn{Dixon split} of $J^{\Vmu}$. It is the
first step towards the extraction of the components.

\begin{lemma}
\label{lm_Split_J_on_Nabla}
\begin{equation}
\begin{aligned}
J^{\Vmu}_{(k)}[\nabla^r_R\,\phi_{\Vmu}]
=
\begin{cases}
\displaystyle
\frac{k!}{(k-r)!}\ 
J^{\Vmu}_{(k)}[\phi_{\Vmu}]
&\quadtext{for}
r \le k
\\
0
&\quadtext{for}
r>k
\DEfullstop
\end{cases}
\end{aligned}
\label{Split_J_on_Nabla}
\end{equation}
\end{lemma}

\begin{proof}

This follows from two observations
\begin{align}
\nabla^{r+1}_R\phi_\Vmu 
= 
-r\nabla^{r}_R\phi_\Vmu + \nabla^{r}_R(\nabla_R\phi_\Vmu)
\label{Split_J_k_induction}
\DEcomma
\end{align}
and
\begin{align}
 J^{\Vmu}_{(k)}[\nabla_R\phi]
=
k\, J^{\Vmu}_{(k)}[\phi]
\label{Split_J_k_Del_phi}
\DEfullstop
\end{align}
Equation (\ref{Split_J_k_induction}) 
follows since in the Dixon \TSchange{adapted} coordinates.
\begin{align*}
\nabla^{r}_R(\nabla_R\phi_\Vmu) & = \nabla^{r}_R(z^b \nabla_b \phi_\Vmu) 
\\
& = z^{a_1} \dots z^{a_r} \nabla_{a_1} \dots \nabla_{a_r} (z^b \nabla_b \phi_\Vmu) 
\\
& = z^{a_1} \dots z^{a_r} \nabla_{a_1} \dots \nabla_{a_{r-1}} 
(\delta^{b}_{a_r} \nabla_b \phi_\Vmu + z^b \nabla_{a_r} \nabla_b \phi_\Vmu) 
\\
& = z^{a_1} \dots z^{a_r} \nabla_{a_1} \dots \nabla_{a_{r-1}} 
(\nabla_{a_r} \phi_\Vmu + z^b \nabla_{a_r} \nabla_b \phi_\Vmu) 
\\
& = 
z^{a_1} \dots z^{a_r} \nabla_{a_1} \dots \nabla_{a_r} 
\phi_\Vmu 
+
z^{a_1} \dots z^{a_r} \nabla_{a_1} \dots \nabla_{a_{r-1}} 
(z^b \nabla_{a_r} \nabla_b \phi_\Vmu) 
\\
& = 
\nabla^{r}_R \phi_\Vmu
+
z^{a_1} \dots z^{a_r} \nabla_{a_1} \dots \nabla_{a_{r-1}} 
(z^b \nabla_{a_r} \nabla_b \phi_\Vmu) 
\\
& = r\nabla^{r}_R\phi_\Vmu + z^{a_1} \dots z^{a_r}z^{b} \nabla_{a_1} \dots \nabla_{a_r}\nabla_{b} \phi_\Vmu 
\\
& = r\nabla^{r}_R\phi_\Vmu + \nabla^{r+1}_R\phi_\Vmu 
\DEfullstop
\end{align*}
Equation (\ref{Split_J_k_Del_phi}) follows from
\begin{align*}
J_{(k)}^{\Vmu}[\nabla_R\,\phi_\Vmu] 
& = 
J_{(k)}^{\Vmu} [z^{b} \nabla_{b} \phi_\Vmu]
= (-1)^{k}\frac{1}{k!}
\int \zeta^{\Vmu a_{1} \dots a_{k}}
\nabla^k_{a_1\cdots a_{k}} 
\Big( z^{b} \nabla_{b} \phi_\Vmu \Big) 
\,d\sigma
\\
& = (-1)^{k}\frac{1}{k!}
\int \zeta^{\Vmu a_{1} \dots a_{k}}\nabla_{a_1} \dots \nabla_{a_{k}} 
\Big( z^{b} \nabla_{b} \phi_\Vmu \Big) 
\,d\sigma
\\
& = (-1)^{k}\frac{1}{k!}
\int \zeta^{\Vmu a_{1} \dots a_{k}}\nabla_{a_1} \dots \nabla_{a_{k-1}} 
\Big(\delta_{a_k}^{b} \nabla_{b} \phi_\Vmu + z^{b} \nabla_{a_k} \nabla_{b} \phi_\Vmu \Big) 
\,d\sigma
\\
& = (-1)^{k}\frac{1}{k!}
\int \zeta^{\Vmu a_{1} \dots a_{k}}\nabla_{a_1} \dots \nabla_{a_{k-1}} 
\Big(\nabla_{a_k} \phi_\Vmu + z^{b} \nabla_{a_k} \nabla_{b} \phi_\Vmu \Big) 
\\
& = (-1)^{k}\frac{1}{k!}
\int \zeta^{\Vmu a_{1} \dots a_{k}}\Big( \nabla_{a_1} \dots \nabla_{a_{k}} \phi_\Vmu
+ \nabla_{a_1} \dots \nabla_{a_{k-1}} (z^{b} \nabla_{a_k} \nabla_{b}
\phi_\Vmu)  
\Big) 
\,d\sigma
\\
& = (-1)^{k}\frac{1}{k!}
\int \zeta^{\Vmu a_{1} \dots a_{k}}\Big(k \nabla_{a_1} \dots \nabla_{a_{k}} \phi_\Vmu
+ z^{b}\nabla_{a_1} \dots \nabla_{a_{k}}\nabla_{b} \phi_\Vmu \Big) 
\,d\sigma
\\
& = (-1)^{k}\frac{1}{(k-1)!}
\int \zeta^{\Vmu a_{1} \dots a_{k}}\nabla_{a_1} \dots \nabla_{a_{k}} \phi_\Vmu 
\,d\sigma
\\
& = (-1)^{k}\frac{1}{(k-1)!}
\int \zeta^{\Vmu a_{1} \dots a_{k}}\nabla_{a_1 \cdots a_{k}} \phi_\Vmu 
\,d\sigma
\\
& = k J_{(k)}^{\Vmu}[\phi_\Vmu]
\DEfullstop
\end{align*}

We now show (\ref{Split_J_on_Nabla}) by induction. Trivial for
$r=0$. Assume true for $r$, from
(\ref{Split_J_k_induction}) and (\ref{Split_J_k_Del_phi})
\begin{align*}
J^{\Vmu}_{(k)}[\nabla^{r+1}_R\,\phi_{\Vmu}]
&=
-r J^{\Vmu}_{(k)}[\nabla^{r}_R\phi_{\Vmu}] + J^{\Vmu}_{(k)}[\nabla^{r}_R(\nabla_R\phi_{\Vmu})]
=
\frac{k!}{(k-r)!}
\big(-r J^{\Vmu}_{(k)}[\phi_{\Vmu}] + J^{\Vmu}_{(k)}[\nabla_R\phi_{\Vmu}]\big)
\\&=
\frac{k!}{(k-r)!}
\big(-r J^{\Vmu}_{(k)}[\phi_{\Vmu}] + k J^{\Vmu}_{(k)}[\phi_{\Vmu}]\big)
=
\frac{k!(k-r)}{(k-r)!}
J^{\Vmu}_{(k)}[\phi_{\Vmu}]
=
\frac{k!}{(k-(r+1))!}
J^{\Vmu}_{(k)}[\phi_{\Vmu}]
\DEfullstop
\end{align*}

In the case when $r=k+1$ then
\begin{align*}
J^{\Vmu}_{(k)}[\nabla^{k+1}_R\,\phi_{\Vmu}]
&=
\frac{k!(k-k)}{(k-r)!}
J^{\Vmu}_{(k)}[\phi_{\Vmu}]
=
0
\DEfullstop
\end{align*}
Hence (\ref{Split_J_on_Nabla}) holds for all $r$.
\end{proof}

\begin{theorem}
\label{thm_Split_Jr}
\begin{align}
J^{\Vmu}_{(r)}[\phi_{\Vmu}] 
= 
 {\sum_{k=r}^{n}}\frac{(-1)^{k-r}}{(k-r)!\ {r!}}  
J^{\Vmu}[\nabla_{R}^{k} \phi_{\Vmu}]
\label{J_J_Jk}
\DEfullstop
\end{align}
\end{theorem}
\begin{proof}

\begin{align*}
 {\sum_{k=r}^{n}}\frac{(-1)^{k-r}}{(k-r)!\ {r!}}  
J^{\Vmu}[\nabla_{R}^{k} \phi_{\Vmu}]
&=
{\sum_{k=r}^{n}}\frac{(-1)^{k-r}}{(k-r)!\ {r!}}  
\sum_{\ell=0}^n
J^{\Vmu}_{(\ell)}[\nabla_{R}^{k} \phi_{\Vmu}]
=
{\sum_{k=r}^{n}}
\sum_{\ell=0}^n
\frac{(-1)^{k-r}}{(k-r)!\ {r!}}  
\frac{\ell!}{(\ell-k)!}
J^{\Vmu}_{(\ell)}[\phi_{\Vmu}]
\\&=
\sum_{\ell=0}^n
{\sum_{k=r}^{n}}
(-1)^{k-r}
\binom{k}{r}
\binom{\ell}{k}
J^{\Vmu}_{(\ell)}[\phi_{\Vmu}]
=
\sum_{\ell=0}^n
\Bigg({\sum_{k=0}^{n}}
(-1)^{k-r}
\binom{k}{r}
\binom{\ell}{k}
\Bigg)
J^{\Vmu}_{(\ell)}[\phi_{\Vmu}]
\\&= 
\sum_{\ell=0}^n \delta_{r\ell}
J^{\Vmu}_{(\ell)}[\phi_{\Vmu}]
=
J^{\Vmu}_{(r)}[\phi_{\Vmu}]
\DEcomma
\end{align*}
where we have used the standard result for the sums of products of
binomial coefficients.
\end{proof}

Although the Dixon split, given by (\ref{Split_arb_J}), is defined via
the coordinate system, theorem \ref{thm_Split_Jr} shows that the
Dixon split is actually independent of the coordinate system, once the
Dixon vector is chosen.

\begin{corrol}
\label{lm_J_0}
If $J^{\Vmu}=0$ then $J^{\Vmu}[\nabla^r_R\phi_{\Vmu}]=0$ for all $r$.
\end{corrol}
\begin{proof}
Follows trivially since if $J^{\Vmu}=0$, then $J^{\Vmu}_{(k)}=0$ 
for all $0\le k\le n$.
\end{proof}

\subsection{Extraction of the Dixon components}
\label{ch_ExtrComp}

In this section we show how we can extract the Dixon components
$\zeta^{\Vmu a_{1} \dots a_{k}}$, in an adapted coordinate system, by
applying the distribution to particular test tensors.

Let $\psi_0:\Real\to\Real$ be any test function such that 
\begin{align*}
\int_\Real \psi_0(z)\,dz=1
,\qquad
\psi_0(0)=1
\DEcomma
\end{align*}
and $\psi_0(z)$ is flat in an interval about $0$. 

In Dixon adapted coordinates, given $\sigma_0\in\Interval$ and
$\epsilon>0$, choose a set of indices $\Vnu$. Let the tensor
$(\phi_{\sigma_0,\epsilon}^{\Vnu,a_1\cdots a_r})_\Vmu$ be given by
\begin{align}
(\phi_{\sigma_0,\epsilon}^{\Vnu,a_1\cdots a_r})_\Vmu|_{(\sigma,z^1,\ldots,z^3)}
=
\frac{(-1)^k}{\epsilon} z^{a_1}\cdots z^{a_r}
\psi_0\Big(\frac{\sigma-\sigma_0}{\epsilon}\Big)
\psi_0(z^1)\psi_0(z^2)\psi_0(z^3)
\PiIField^{\FixInd\Vnu}_{\Vmu}
\label{Split_phi_extract}
\DEfullstop
\end{align}

\begin{theorem}
\label{thm_ExtrComp_zeta}
\begin{align}
\zeta^{\Vnu a_1\cdots a_k}(\sigma_0)
=
\lim_{\epsilon\to 0}
J^{\Vmu}_{(k)}[
(\phi_{\sigma_0,\epsilon}^{\Vnu,a_1\cdots a_r})_\Vmu 
]
\label{Split_zeta_extract}
\DEfullstop
\end{align}
\end{theorem}
\begin{proof}
Since $\psi_0(z^a)$ is flat about $z^a=0$, then 
$\nabla_b\psi_0(z^a)=0$. Since $\zeta^{\Vnu a_1\cdots a_k}$ is totally
symmetric in $a_1\ldots a_k$ we have from (\ref{Prll_multideriv_Pi})
$\nabla^k_{a_1\cdots a_k} \PiIField^{\FixInd\Vnu}_\Vmu = 0$. Also
$\nabla_b z^a=\partial_b z^a=\delta^a_b$.
\begin{align*}
\lim_{\epsilon \to 0} J^{\Vmu}_{(k)} &
[(\phi_{\sigma_0,\epsilon}^{\Vnu,a_1\cdots a_r})_\Vmu] 
= 
\lim_{\epsilon \to 0} 
\frac{1}{k!\,\epsilon} 
\int_\Interval \zeta^{\Vmu b_1\cdots b_k} 
\nabla^k_{b_1\cdots b_k}
\bigg( z^{a_1}\cdots z^{a_k}
\psi_0\Big(\frac{\sigma-\sigma_0}{\epsilon}\Big)
\psi_0(z^1)\psi_0(z^2)\psi_0(z^3)
\PiIField^{\FixInd\Vnu}_{\Vmu}\, d\sigma 
\bigg)
\\&=
\lim_{\epsilon \to 0} 
\frac{1}{k!\,\epsilon} 
\int_\Interval \zeta^{\Vmu b_1\cdots b_k} 
\nabla^k_{b_1\cdots b_k}
\big( z^{a_1}\cdots z^{a_k}\big)
\psi_0\Big(\frac{\sigma-\sigma_0}{\epsilon}\Big)
\psi_0(z^1)\psi_0(z^2)\psi_0(z^3)
\PiIField^{\FixInd\Vnu}_{\Vmu}\big|_{z=0}\, d\sigma 
\\&=
\lim_{\epsilon \to 0} 
\frac{1}{k!\,\epsilon} 
\int_\Interval \zeta^{\Vmu b_1\cdots b_k} 
\partial_{b_1}\ldots\partial_{b_k}
\big( z^{a_1}\cdots z^{a_k}\big)
\psi_0\Big(\frac{\sigma-\sigma_0}{\epsilon}\Big)
\psi_0(0)^3
\delta^{\Vnu}_{\Vmu}\big|_{z=0}\, d\sigma 
\\&=
\lim_{\epsilon \to 0} 
\frac{1}{\epsilon} 
\int_\Interval \zeta^{\Vnu a_1\cdots a_k} 
\psi_0\Big(\frac{\sigma-\sigma_0}{\epsilon}\Big)
\, d\sigma 
=
\zeta^{\Vnu a_1\cdots a_k} (\sigma_0)
\DEcomma
\end{align*}
as $\epsilon^{-1}\psi_0(\epsilon^{-1}(\sigma-\sigma_0))\to
\delta(\sigma-\sigma_0)$ as $\epsilon\to0$.
\end{proof}

\begin{corrol}
\label{thm_ExtrComp_unique}
\begin{align}
J^{\Vmu}=0 
\qquadtext{if and only if}
\zeta^{\Vmu a_1\cdots a_k} = 0
\label{J_zeta_unique}
\DEcomma
\end{align}
and hence the Dixon components are unique.
\end{corrol}
\begin{proof}
Assuming $J^{\Vmu}=0$ then from corollary \ref{lm_J_0} we have
$J^{\Vmu}_{(k)}=0$, hence from (\ref{Split_zeta_extract}) all the
components are zero.
\end{proof}


\subsection{The Taylor expansion \TSchange{with respect to} Dixon geodesic hypersurfaces}
\label{ch_Tay}
Given a test tensor $\phi_{\Vmu}$, we need to take its Taylor
expansion, in a way that respects the Dixon geodesic hypersurfaces. 
We define the tensor field
\begin{align}
\phi^{(k)}_\Vmu|_{P(\sigma,V)}
&=
\Pi^\Vnu_\Vmu\big|_{C(\sigma)}^{P(\sigma,V)} (\nabla_{V}^k \phi_\Vnu|_{C(\sigma)})
\label{Tay_def_phik}
\DEcomma
\end{align}
by replacing $V$ with $sV$ this becomes
\begin{align}
\phi^{(k)}_\Vmu|_{P(\sigma,s V)}
&=
\Pi^\Vnu_\Vmu\big|_{C(\sigma)}^{P(\sigma,s V)} (\nabla_{s V}^k \phi_\Vnu|_{C(\sigma)})
=
s^k\, \Pi^\Vnu_\Vmu\big|_{C(\sigma)}^{P(\sigma,s V)} (\nabla_{V}^k \phi_\Vnu|_{C(\sigma)})
\label{Tay_phik_res1}
\DEfullstop
\end{align}
In \TSchange{adapted} coordinates $H_{\sigma,V}(1) = P(\sigma,V) =
(\sigma,z^1,\ldots,z^3)$ we have
\begin{align}
\phi^{(k)}_\Vmu |_{(\sigma,z^1,\ldots,z^3)}
=
z^{a_1}\ldots z^{a_k}
\PiIField^{\FixInd\Vnu}_\Vmu 
\big(\nabla_{a_1\cdots a_k}^k \phi_\Vnu|_{C(\sigma)}\big)
\label{Tay_phik_adap}
\DEcomma
\end{align}
since
\begin{align*}
\phi^{(k)}_\Vmu |_{(\sigma,z^1,\ldots,z^3)}
&=
\phi^{(k)}_\Vmu |_{P(\sigma,V)}
=
\Pi^\Vnu_\Vmu\big|_{C(\sigma)}^{P(\sigma,V)} \nabla_{{V}}^k \phi_\Vnu|_{C(\sigma)}
=
\PiIField^{\FixInd\Vnu}_\Vmu \big(\nabla_{{V}}^k \phi_\Vnu|_{C(\sigma)}\big)
=
V^{a_1}\ldots V^{a_k}
\PiIField^{\FixInd\Vnu}_\Vmu \big(\nabla_{a_1\cdots a_k}^k \phi_\Vnu|_{C(\sigma)}\big)
\\&=
z^{a_1}\ldots z^{a_k}
\PiIField^{\FixInd\Vnu}_\Vmu \big(\nabla_{a_1\cdots a_k}^k \phi_\Vnu|_{C(\sigma)}\big)
\DEfullstop
\end{align*}
The Taylor expansion, \TSchange{with respect to} the Dixon geodesic hypersurfaces, of a test
tensor $\phi_\Vmu$ is given by 
\begin{align}
\phi_\Vmu|_{P(\sigma,s V)}
&=
\sum_{k=0}^{N} 
\frac{1}{k!}\,
\phi^{(k)}_\Vmu|_{P(\sigma,s V)}
+O(s^{N+1})
\label{Tay_phi_sum}
\DEfullstop
\end{align}
\begin{lemma}
The order $O(s^{N+1})$, as $s\to0$, in \eqref{Tay_phi_sum} is correct.
\end{lemma}
\begin{proof}
Taking the Taylor expansion about $C(\sigma)$ along each geodesic
$H_{\sigma,V}(s)$ and using (\ref{Prll_higher_deriv}) we have
\begin{align*}
\Pi^\Vmu_\Vnu\big|^{C(\sigma)}_{P(\sigma,s V)} \phi_\Vmu|_{P(\sigma,s V)}
&=
\sum_{k=0}^{N} 
\frac{s^k}{k!}\,
\frac{d^k}{ds^k}\Big|_{s=0} 
\Big(\Pi^\Vmu_\Vnu\big|^{C(\sigma)}_{P(\sigma,s  V)} 
\phi_\Vmu|_{P(\sigma,s V)}\Big)
+O(s^{N+1})
\\&=
\sum_{k=0}^{N} 
\frac{s^k}{k!}\,
\nabla_{{V}}^k \phi_\Vnu|_{C(\sigma)}
+O(s^{N+1})
\DEfullstop
\end{align*}
Hence
\begin{align*}
\phi_\Vmu|_{P(\sigma,s V)}
&=
\sum_{k=0}^{N} 
\frac{s^k}{k!}\, \Pi^\Vnu_\Vmu\big|_{C(\sigma)}^{P(\sigma,s V)} \nabla_{{V}}^k \phi_\Vnu|_{C(\sigma)}
+O(s^{N+1})
=
\sum_{k=0}^{N} 
\frac{1}{k!}\,
\phi^{(k)}_\Vmu|_{P(\sigma,s V)}
+O(s^{N+1})
\DEfullstop
\end{align*}
\end{proof}

\begin{lemma}
\begin{align}
\nabla^r_R
\phi^{(k)}_\Vmu
=
\frac{k!}{(k-r)!}\, \phi^{(k)}_\Vmu
\label{Tay_nabla_r_phi_k=0}
\DEfullstop
\end{align}
\end{lemma}
\begin{proof}
From (\ref{Prll_multideriv_Pi}) and (\ref{Tay_phik_adap})
\TSchange{
\begin{align*}
\nabla^r_R \phi^{(k)}_\Vmu 
&=
z^{b_1}\ldots z^{b_r} \nabla^r_{b_1\cdots b_r} \phi^{(k)}_\Vmu 
\\&=
z^{b_1}\ldots z^{b_r} \nabla^r_{b_1\cdots b_r} \Big(z^{a_1}\ldots z^{a_k}
\PiIField^{\FixInd\Vnu}_\Vmu 
\big(\nabla_{a_1\cdots a_k}^k \phi_\Vnu|_{C(\sigma)}\big)\Big)
\\&=
z^{b_1}\cdots z^{b_r}
\partial_{b_1}\ldots \partial_{b_r} \big(z^{a_1}\ldots z^{a_k}\big)
\PiIField^{\FixInd\Vnu}_\Vmu 
\big(\nabla_{a_1\cdots a_k}^k \phi_\Vnu|_{C(\sigma)}\big)
\\&=
\frac{k!}{(k-r)!}\, \big(z^{a_1}\ldots z^{a_k}\big)
\PiIField^{\FixInd\Vnu}_\Vmu 
\big(\nabla_{a_1\cdots a_k}^k \phi_\Vnu|_{C(\sigma)}\big)
=
\frac{k!}{(k-r)!}\, \phi^{(k)}_\Vmu
\DEfullstop
\end{align*}
}
\end{proof}

\subsection{All multipoles are Dixon multipoles}
\label{ch_Span}

In the previous subsection we showed that the components of a Dixon
multipole are unique and can be extracted using particular text
functions given by (\ref{J_J_Jk}).  In this section we show that the
all multipoles can be written as a Dixon multipole. Thus we see that
the Dixon multipole is merely a representation of a multipole.  Thus
if a multipole is defined via (\ref{J_arb_J_form}) with respect to one
Dixon vector, $\DixVec_\mu$, then it is guaranteed that it can be
written with respect to another Dixon vector $\hat\DixVec_\mu$.  A
similar result is available in \cite{pinto2022thesis} but here an
explicit formula is given, which in addition respects the constraint
(\ref{Intro_Tab_Dixon_N_NonOrth}).

We assume we are given an arbitrary multipole $\Jalt^{\Vmu}$, i.e. an operator
which takes test tensors $\phi_{\Vmu}$ to give a number
$\Jalt^{\Vmu}[\phi_{\Vmu}]$. In \cite{gratus2020distributional} we show how to construct an
arbitrary distribution. This is by taking a collection of monopoles,
and allowing finite sums, derivatives, contractions and products with
scalar fields. In \cite{gratus2020distributional} we showed that all multipoles could be
expressed as Ellis multipoles. Thus another interpretation is that
$\Jalt^{\Vmu}$ is given in the Ellis representation. An alternative
is that $\Jalt^{\Vmu}$ is given to us in \TSchange{the} Dixon representation, but
with a different Dixon vector $\DixVec_\mu$. Even in this case it is
not obvious that the multipole can be represented as a Dixon multipole
with the new Dixon vector.

Our goal is to establish $\Jalt^{\Vmu}$ can be
written as a Dixon multipole, but the only information we can use are
the values $\Jalt^{\Vmu}[\phi_{\Vmu}]$ for particular $\phi_{\Vmu}$'s.

The result that all multipoles can be written as Dixon multipoles,
should not be surprising as if one \TSchange{were} to count the number of
components in the Ellis representation, it is the same as the
Dixon representation. However one should be careful, as the components
in the Ellis representation, also have the same number of components,
but they are not unique. So simply counting number
of components is not sufficient.

We can deduce the order of $\Jalt^{\Vmu}$ by
the statement \cite[eqn. (115)]{gratus2020distributional}: The order of $\Jalt^{\Vmu}$ is the smallest $n$
such that
\begin{align}
\Jalt^{\Vmu}[\lambda^{n+1}\phi_{\Vmu}]=0
\quad\text{for all tensors $\phi_{\Vmu}$ and all 
scalar fields $\lambda$ such that $\lambda|_{C(\sigma)}=0$.}
\label{Span_def_order}
\end{align}
Having established the order of $\Jalt^{\Vmu}$ we can now perform the
Dixon split. I.e. using (\ref{J_J_Jk}) we can create the multipoles 
$\Jalt^{\Vmu}_{(k)}$
\begin{align}
\Jalt^{\Vmu}_{(r)}[\phi_{\Vmu}] 
= 
 {\sum_{k=r}^{n}}\frac{(-1)^{k-r}}{(k-r)!\ {r!}}  
\Jalt^{\Vmu}[\nabla_{R}^{k} \phi_{\Vmu}]
\label{Span_J_Jk}
\DEfullstop
\end{align}

We can now use (\ref{Split_zeta_extract}) to calculate the components
of $\Jalt^{\Vmu}$,
\begin{align}
\zetaalt^{\Vnu a_1\cdots a_k}(\sigma_0)
=
\lim_{\epsilon\to 0}
\Jalt^{\Vmu}_{(k)}[
(\phi_{\sigma_0,\epsilon}^{\Vnu,a_1\cdots a_r})_\Vmu 
]
\label{Span_zetaalt}
\DEfullstop
\end{align}
We emphasise here that $\Jalt^{\Vmu}$ was not
defined using $\zetaalt^{\Vnu a_1\cdots a_k}$, via (\ref{J_arb_J_form}),
instead $\Jalt^{\Vmu}$ was given to us and we have used various test
functions to calculate the components.

It is now necessary to establish that (\ref{J_arb_J_form}) gives us
the correct distribution. To do this we define a new distribution as
\begin{align}
J^{\Vmu}
= 
\sum_{k=0}^{N} 
\frac{1}{k!}\, \nabla^k_{\rho_1\ldots\rho_k}
\int_\Interval \zetaalt^{\Vmu\rho_1\cdots\rho_k} 
\deltafour\big(x-C(\sigma)\big)\,d\sigma
\label{Span_Jalt}
\DEfullstop
\end{align}
It is then necessary to show $J^{\Vmu}=\Jalt^{\Vmu}$.

\begin{theorem}
\label{thm_Span}
\begin{align}
J^{\Vmu} = \Jalt^{\Vmu}
\label{Span_J=Jalt}
\DEfullstop
\end{align}
Hence all all
multipoles over $C$ are Dixon multipoles over $C$.
\end{theorem}
\begin{proof}
\begin{align*}
J^{\Vmu}_{(k)}[\phi_\Vmu]
&=
\frac{(-1)^k}{k!} \int_\Interval 
\zetaalt^{\Vmu a_1\cdots a_k}(\sigma_0) 
(\nabla_{a_1\cdots a_k}^k \phi_\Vnu)|_{C(\sigma_0)}\ 
d\sigma_0
\\&=
\frac{(-1)^k}{k!} \int_\Interval 
\lim_{\epsilon\to 0}
\Jalt^{\Vmu}_{(k)}[
(\phi_{\sigma_0,\epsilon}^{\Vnu,a_1\cdots a_r})_\Vmu 
]
(\nabla_{a_1\cdots a_k}^k \phi_\Vnu)|_{C(\sigma_0)}\ 
d\sigma_0
\\&=
\frac{(-1)^k}{k!} \lim_{\epsilon\to 0}
\Jalt^{\Vmu}_{(k)}\Big[
\int_\Interval 
(\phi_{\sigma_0,\epsilon}^{\Vnu,a_1\cdots a_r})_\Vmu 
(\nabla_{a_1\cdots a_k}^k \phi_\Vnu)|_{C(\sigma_0)}\ 
d\sigma_0
\Big]
\\&=
\frac{1}{k!}\lim_{\epsilon\to 0}
\Jalt^{\Vmu}_{(k)}\Big[
\int_\Interval 
\frac{1}{\epsilon} z^{a_1}\cdots z^{a_r}
\Big(\psi_0\Big(\frac{\sigma-\sigma_0}{\epsilon}\Big)
\psi_0(z^1)\psi_0(z^2)\psi_0(z^3)
\PiIField^{\FixInd\Vnu}_{\Vmu}\Big) 
(\nabla_{a_1\cdots a_k}^k \phi_\Vnu)|_{C(\sigma_0)}\ 
d\sigma_0
\Big]
\\&=
\lim_{\epsilon\to 0}
\frac{1}{k!\,\epsilon} 
\Jalt^{\Vmu}_{(k)}\Big[
\int_\Interval \psi_0\Big(\frac{\sigma-\sigma_0}{\epsilon}\Big)
\psi_0(z^1)\psi_0(z^2)\psi_0(z^3)
z^{a_1}\cdots z^{a_r} \PiIField^{\FixInd\Vnu}_{\Vmu}
(\nabla_{a_1\cdots a_k}^k \phi_\Vnu)|_{C(\sigma_0)}\ 
d\sigma_0
\Big]
\\&=
\lim_{\epsilon\to 0}
\frac{1}{k!\,\epsilon} 
\Jalt^{\Vmu}_{(k)}\Big[
\int_\Interval \psi_0\Big(\frac{\sigma-\sigma_0}{\epsilon}\Big)
\psi_0(z^1)\psi_0(z^2)\psi_0(z^3)
\phi^{(k)}_\Vmu|_{C(\sigma_0)}
d\sigma_0
\Big]
\\&=
\lim_{\epsilon\to 0}
\frac{1}{k!\,\epsilon} 
\Jalt^{\Vmu}_{(k)}\Big[
\int_\Interval \psi_0\Big(\frac{\sigma-\sigma_0}{\epsilon}\Big)
\phi^{(k)}_\Vmu|_{C(\sigma_0)}
d\sigma_0
\Big]
\\&=
\frac{1}{k!} 
\Jalt^{\Vmu}_{(k)}\Big[
\lim_{\epsilon\to 0} \epsilon^{-1}
\int_\Interval \psi_0\Big(\frac{\sigma-\sigma_0}{\epsilon}\Big)
\phi^{(k)}_\Vmu|_{C(\sigma_0)}
d\sigma_0
\Big]
\\&=
\frac{1}{k!} 
\Jalt^{\Vmu}_{(k)}[
\phi^{(k)}_\Vmu
]
=
\frac{1}{k!} 
 {\sum_{\ell=k}^{n}}\frac{(-1)^{\ell-k}}{(\ell-k)!\ {k!}}  
\Jalt^{\Vmu}[\nabla_{R}^{\ell} \phi^{(k)}_{\Vmu}]
\\&=
\frac{1}{k!} 
 {\sum_{\ell=k}^{k}}\frac{(-1)^{\ell-k}}{(\ell-k)!\ {k!}}\,
\frac{k!}{(k-\ell)!}\, 
\Jalt^{\Vmu}[\phi^{(k)}_{\Vmu}]
=
\frac{1}{k!} 
\Jalt^{\Vmu}[\phi^{(k)}_{\Vmu}]
\DEfullstop
\end{align*}
Hence
\begin{align*}
J^{\Vmu}[\phi_\Vmu]
=
\sum_{k=0}^N 
J^{\Vmu}_{(k)}[\phi_\Vmu]
=
\sum_{k=0}^N
\frac{1}{k!} 
\Jalt^{\Vmu}\Big[
\phi^{(k)}_\Vmu
\Big]
=
\Jalt^{\Vmu}\Big[
\sum_{k=0}^N
\frac{1}{k!} 
\phi^{(k)}_\Vmu
\Big]
=
\Jalt^{\Vmu}[\phi_\Vmu]
\DEfullstop
\end{align*}

Since for any arbitrary distribution $\Jalt^{\Vmu}$ 
we can construct using (\ref{Span_J_Jk}), (\ref{Span_zetaalt}) and
(\ref{Span_Jalt}) $J^{\Vmu}$, then  (\ref{Span_J=Jalt}) implies all
multipoles over $C$ are Dixon multipoles over $C$.

Observe that between lines 7 and 8 of this derivation we swapped
$\lim_{\epsilon\to 0}$ and $\Jalt^{\Vmu}_{(k)}$. This is permitted
since $\Jalt^{\Vmu}_{(k)}$ consists of only a finite number of
derivatives and is
hence continuous and
\begin{align*}
\lim_{\epsilon\to 0} \epsilon^{-1}
\int_\Interval \psi_0\Big(\frac{\sigma-\sigma_0}{\epsilon}\Big)
\phi^{(k)}_\Vmu|_{C(\sigma_0)}
d\sigma_0
\to
\phi^{(k)}_\Vmu|_{C(\sigma)}
\DEcomma
\end{align*}
in the  LF-topology.
\end{proof}

\begin{corrol}
Given a multipole $J^{\Vmu}$ defined by \eqref{J_arb_J_form} with respect to one
Dixon vector $\DixVec_\mu$, and given a different Dixon vector
$\hat\DixVec_\mu$, then there exist unique components
$\hat{\zeta}^{\Vmu a_1\cdots a_k}$ so that we can write $J^{\Vmu}$
with respect to $\hat\DixVec_\mu$.
\end{corrol}
\begin{proof}
Using theorem \ref{thm_Span} and corollary \ref{thm_ExtrComp_unique}.
\end{proof}

\subsection{Parallel squeezed tensors}
\label{ch_Sqz}

So far, moments have been discussed as components of the
distribution. There is also the intuitive notion of moments in terms
of integrals over space, multiplied by one or more coordinates.
 
In general relativity such objects are not covariant as they depend on both
the choice of the spatial hypersurface and the choice of the
coordinates on that hypersurface. The other issue is that
the tensor \TSchange{has} to be integrated, and this requires
transporting the tensors to the same point. 

One choice, given in \cite{gratus2020distributional} is to use an adapted coordinate system. This
is coordinate dependent. However by squeezing the tensor field, one
can produce a well defined distribution. This \TSchange{corresponds} to the
Ellis representation of the distribution.

Dixon \cite{DixonII} \TSchange{gives} a natural choice of moments. This uses the
\TSchange{Dixon geodesic hypersurfaces}, the geodesic coordinates and parallel
transport. Thus all the necessary structure is given to us uniquely,
once we have chosen the Dixon vector.

In this section we give a method of squeezing a tensor field, such
that the coefficients of the expansion are the Dixon components of the
distribution.

Let $U^\Vmu$ be a tensor density (of weight 1) on $\Mman$. From this we
can construct the squeezed tensor density
$U_\epsilon^\Vmu \in \Gamma TM$ given by
\begin{align}
U_\epsilon^\Vmu |_{P(\sigma,V)}
=
\epsilon^{-3}\,\Pi^\Vmu_\Vnu\big|^{P(\sigma,\epsilon^{-1}V}_{P(\sigma,V)}
\,U^\Vnu|_{P(\sigma,\epsilon^{-1}V)}
\label{Sqz_Dix_def_U_eps_Mman}
\DEfullstop
\end{align}
Let
\begin{align}
\xi^{\Vmu \rho_1\cdots \rho_k}(\sigma)
=
(-1)^k \int_{\Nspace(\sigma)}\, V^{\rho_1}\cdots V^{\rho_k}\,
\Pi^\Vmu_\Vnu\big|^{P(\sigma,V)}_{C(\sigma)}\
 U^\Vnu|_{P(\sigma,V)}\, d^3 V
\label{Prll_def_xi_N}
\DEfullstop
\end{align}
These components clearly satisfy (\ref{J_arb_J_form}) and (\ref{J_DixVec_Cons}).
In Dixon geodesics coordinates this becomes
\TSchange{
\begin{align}
\xi^{\Vmu \rho_1\cdots \rho_k}(\sigma)
=
(-1)^k \int_{\Sigma(\sigma)}\, z^{\rho_1}\cdots z^{\rho_k}\,
 U^\Vmu|_{(\sigma,z^1,z^2,z^3)}\, d^3 z
\label{Prll_def_xi_Sig}
\DEfullstop
\end{align}
}

\begin{lemma}
\begin{align}
\int_\Mman U_\epsilon^\Vmu\, \phi_\Vmu\, d^4x
&=
\sum_{k=0}^N
(-1)^k \frac{\epsilon^k}{k!}
\int_\Interval
\, \xi^{\Vmu \rho_1\cdots \rho_k}\,
\nabla^k_{\rho_1\cdots \rho_k}\phi_\Vmu 
\, d\sigma
+O(\epsilon^{N+1})
\label{Prll_expansion_Sig}
\DEfullstop
\end{align}
Thus
\TSchange{
\begin{align}
U^\Vmu_\epsilon
=
\sum_{k=0}^N
\frac{\epsilon^k}{k!}
\nabla^k_{\rho_1\cdots\rho_k}
\int_\Interval d\sigma \,
\xi^{\Vmu \rho_1\cdots \rho_k}
\,
\deltaFour (x-C)
+O(\epsilon^{N+1})
\label{Prll_expansion_delta}
\DEfullstop
\end{align}
}
\end{lemma}

\begin{proof}
From (\ref{Tay_phik_res1}) and (\ref{Tay_phi_sum}) we have
\begin{align*}
\Pi^\Vmu_\Vnu\big|^{C(\sigma)}_{P(\sigma,s V)}\phi_\Vmu|_{P(\sigma,s V)}
&=
\sum_{k=0}^{N} 
\frac{s^k}{k!}\,
\nabla_{V}^k \phi_\Vnu|_{C(\sigma)}
+O(s^{N+1})
\DEfullstop
\end{align*}
In Dixon geodesic coordinates $(\sigma,\TSchange{z^1},z^2,z^3)$. We set the
vector field $V^\mu=\epsilon\,W^\mu$ and do a Taylor expansion in
$\epsilon$. 
\TSchange{
\begin{align*}
\int_\Mman & U_\epsilon^\Vmu\, \phi_\Vmu\, d\sigma\,d^3 z
\\&=
\int_\Interval d\sigma 
\int_{\Sigma(\sigma)} U_\epsilon^\Vmu\, \phi_\Vmu\, d^3 z
\\&=
\int_\Interval d\sigma 
\int_{\Nspace(\sigma)} 
U_\epsilon^\Vmu|_{P(\sigma,V)}\, \phi_\Vmu|_{P(\sigma,V)}\,
d^3V
\\&=
\epsilon^{-3}
\int_\Interval d\sigma 
\int_{\Nspace(\sigma)} \Pi^\Vmu_\Vnu\big|^{P(\sigma,\epsilon^{-1}V)}_{P(\sigma,V)}
\,U^\Vnu|_{P(\sigma,\epsilon^{-1}V)}\ 
\phi_\Vmu|_{P(\sigma,V)}\, d^3V
\\&=
\epsilon^{-3}
\int_\Interval d\sigma 
\int_{\Nspace(\sigma)} 
\Pi^\Vmu_\Vlambda\big|^{C(\sigma)}_{P(\sigma,V)}\ 
\Pi^\Vlambda_\Vnu\big|^{P(\sigma,\epsilon^{-1}V)}_{C(\sigma)}\ 
\,U^\Vnu|_{P(\sigma,\epsilon^{-1}V)}\ 
\phi_\Vmu|_{P(\sigma,V)}\, d^3V
\\&=
\int_\Interval d\sigma 
\int_{\Nspace(\sigma)} 
\Pi^\Vlambda_\Vnu\big|^{P(\sigma,W)}_{C(\sigma)}
\,U^\Vnu|_{P(\sigma,W)}\, 
\Pi^\Vmu_\Vlambda\big|^{C(\sigma)}_{P(\sigma,\epsilon W)}\ 
\phi_\Vmu|_{P(\sigma,\epsilon W)}\, d^3W
\\&=
\int_\Interval d\sigma 
\int_{\Nspace(\sigma)} 
\Pi^\Vlambda_\Vnu\big|^{P(\sigma,W)}_{C(\sigma)}
\,U^\Vnu|_{P(\sigma,W)}\, 
d^3W
\bigg(
\sum_{k=0}^N
\frac{\epsilon^k}{k!}
\nabla_W^k
\phi_\Vlambda|_{C(\sigma)} 
\bigg)
+O(\epsilon^{N+1})
\\&=
\sum_{k=0}^N
\frac{\epsilon^k}{k!}
\int_\Interval d\sigma 
\int_{\Nspace(\sigma)} 
\Pi^\Vlambda_\Vnu\big|^{P(\sigma,W)}_{C(\sigma)}
\,U^\Vnu|_{P(\sigma,W)}\, 
W^{\rho_1}\cdots W^{\rho_k}
\,
d^3W
\big(
\nabla^k_{\rho_1\cdots\rho_k}
\phi_\Vlambda|_{C(\sigma)} 
\big)
+O(\epsilon^{N+1})
\\&=
\sum_{k=0}^N
(-1)^k\frac{\epsilon^k}{k!}
\int_\Interval d\sigma \,
\xi^{\Vlambda \rho_1\cdots \rho_k}(\sigma)
\big(
\nabla^k_{\rho_1\cdots\rho_k}
\phi_\Vlambda|_{C(\sigma)} 
\big)
+O(\epsilon^{N+1})
\DEfullstop
\end{align*}
}
\end{proof}

Thus we have shown the link between the components of the multipoles
and the moments of a regular distribution.

\section{Dynamical equations in Dixon language}
\label{ch_DivF}

In this section we derive the key result of this article, namely the
dynamical equations of the stress-energy quadrupole, in the Dixon
representation. These follow from the \TSchange{divergenceless} condition
(\ref{SE_Div_Free}). As stated this has the advantage that they
are tensorial. The tensorial expression is given below in theorem
\ref{thm_DixQuad_Ten}. We first derive the equation in Dixon adapted
\TSchange{coordinates}.

\begin{theorem}
In Dixon adapted coordinates the \TSchange{divergenceless} condition
(\ref{SE_Div_Free}) \TSchange{corresponds} to the following dynamical equations for the
components
\begin{align}
\nabla_0\xi^{\mu0\iSb\iSc}
&=
-2\xi^{\mu(\iSb\iSc)}
\label{DivF_xi4dot_adapt}
\DEcomma
\\
\nabla_0\xi^{\mu0\iSa}
&=
- \xi^{\mu a}\,
+\tfrac12
\xi^{\mu0\iSb\iSc}\,
 R^{a}{}_{\iSb 0\iSc} 
+
\xi^{\rho0\iSa\iSc}\,R^{\mu}{}_{\rho0\iSc} 
+
\tfrac12\xi^{\rho c \iSb\iSa}\,R^\mu{}_{\rho c \iSb}
+
\tfrac16 \xi^{\mu d \iSb\iSc}\,R^{a}{}_{\iSb d \iSc} 
\label{DivF_xi3dot_adapt}
\DEcomma
\\
\nabla_0\xi^{\mu 0}
&=
\xi^{\rho0\iSb}\,
R^\mu{}_{\rho0\iSb}\,
+
\tfrac12
\xi^{\rho a \iSb}\,
R^\mu{}_{\rho a \iSb}\,
+\tfrac12\nabla_{0} \big( \xi^{\mu 0\iSb\iSc}\,
R^{0}{}_{\iSb 0\iSc} \big) \,
+\tfrac12
\xi^{\rho0\iSb\iSc}\,
(\nabla_\iSc R^\mu{}_{\rho0\iSb})
+\tfrac16\nabla_{0} \big( \xi^{\mu a \iSb\iSc}\,
R^{0}{}_{\iSb a \iSc} \big) \,
\notag
\\&\qquad
-\tfrac13
\xi^{\rho a \iSb\iSc}\,
(\nabla_\iSc R^\mu{}_{\rho a \iSb})
\label{DivF_xi2dot_adapt}
\DEcomma
\end{align}
together with the constraint
\begin{align}
\xi^{\mu(\iSa\iSb\iSc)} 
&= 0
\label{DivF_xi_constraint_adapt}
\DEfullstop
\end{align}

\end{theorem}

\begin{proof}
For this proof, all integrals are implicitly over $\Interval$.
Let
\begin{align*}
F^\mu = 
\nabla_\sigma T^{\sigma\mu} 
=
F_{(1)}^\mu + 
F_{(2)}^\mu + 
F_{(3)}^\mu 
\DEcomma
\end{align*}
where
\begin{align}
F_{(1)}^\mu
=
\nabla_\nu \int \xi^{\mu\nu} \delta(x-C) d\sigma
,\quad
F_{(2)}^\mu
=
\nabla_\nu \nabla_\rho \int \xi^{\mu\nu\rho} \delta(x-C) d\sigma
,\quad
F_{(3)}^\mu
=
\tfrac12\nabla_\nu \nabla^2_{\rho\sigma}
\int \xi^{\mu\nu\rho\sigma} 
\delta(x-C) d\sigma
\label{DivF_lm_def_f123}
\DEfullstop
\end{align}
Manipulating $F^\mu_{(1)}$, $F^\mu_{(2)}$ and $F^\mu_{(3)}$ in turn.
\begin{align*}
F^\mu_{(1)}[\phi_\mu]
& =
-
\int \xi^{\mu\nu}\,
\nabla_\nu \,
\phi_\mu
\,d\sigma 
= -\int \xi^{\mu a}\,
\nabla_a \,
\phi_\mu \,d\sigma 
-
\int \xi^{\mu 0}\,
\nabla_0 \phi_\mu
\,d\sigma
\\&= -\int \xi^{\mu a}\,
\nabla_a \phi_\mu \,d\sigma 
+ \int (\nabla_0\xi^{\mu 0})\,
\phi_\mu
\,d\sigma
\DEcomma
\end{align*}
which using (\ref{Split_J_on_Nabla}) in lemma \ref{lm_Split_J_on_Nabla} gives
\begin{align}
F^\mu_{(1)}[\nabla^3_R \phi_\mu]
&=
F^\mu_{(1)}[\nabla^2_R \phi_\mu]
=
0
,
\label{DivF_lm_F1_Del2}
\\
F^\mu_{(1)}[\nabla_R \phi_\mu]
&=
-
\int \xi^{\mu a}\,
\nabla_a\phi_\mu \,d\sigma 
,
\label{DivF_lm_F1_Del1}
\\
\qquadand
F^\mu_{(1)}[\phi_\mu]
&=
F^\mu_{(1)}[\nabla_R \phi_\mu]
+
 \int (\nabla_0\xi^{\mu 0})\,
\phi_\mu
\,d\sigma
\label{DivF_lm_F1_Del0}
\DEfullstop
\end{align}

For $F^\mu_{(2)}[\phi_\mu]$ we have
\begin{align*}
F_{(2)}^\mu\big[\phi_\mu\big]
&= 
\int \xi^{\mu\nu\rho}\,
\nabla_\rho\nabla_\nu\,
\phi_\mu
\,d\sigma
=
\int \xi^{\mu0\iSb}\,
\nabla_\iSb\nabla_0\,
\phi_\mu
\,d\sigma
+
\int \xi^{\mu\iSa\iSb}\,
\nabla_\iSb\nabla_\iSa \,
\phi_\mu
\,d\sigma
\\
& = \int \xi^{\mu0\iSb}\,
\nabla_0\nabla_\iSb\,
\phi_\mu
\,d\sigma
-
\int \xi^{\mu0\iSb}\,
R^\rho{}_{\mu0\iSb}\,
\phi_\rho
\,d\sigma 
+\int \xi^{\mu (a \iSb)}\,
\nabla_a \nabla_\iSb\,
\phi_\mu
\,d\sigma
-\tfrac12
\int \xi^{\mu a \iSb}\,
R^\rho{}_{\mu a \iSb}\,
\phi_\rho
\,d\sigma 
\\
& = 
-\int (\nabla_0\xi^{\mu0\iSb})
\nabla_\iSb\,
\phi_\mu
\,d\sigma
-
\int \xi^{\mu0\iSb}\,
R^\rho{}_{\mu0\iSb}\,
\phi_\rho
\,d\sigma 
+\int \xi^{\mu (a \iSb)}\,
\nabla_{ab}
\phi_\mu
\,d\sigma
-\tfrac12
\int \xi^{\mu a \iSb}\,
R^\rho{}_{\mu a \iSb}\,
\phi_\rho
\,d\sigma 
\DEfullstop
\end{align*}
Using equation (\ref{Split_J_on_Nabla}) we can split $F^\mu_{(2)}[\phi_\mu]$
to give
\begin{align}
F_{(2)}^\mu\big[\nabla^3_R\phi_\mu\big] &= 0 
\label{DivF_lm_F2_Del3}
\\
F_{(2)}^\mu\big[\nabla^2_R\phi_\mu\big] 
&=
2 \int \xi^{\mu (a \iSb)}\,
\nabla_{ab}
\phi_\mu
\,d\sigma
\label{DivF_lm_F2_Del2}
\\
F_{(2)}^\mu\big[\nabla_R\phi_\mu\big] 
&= 
\tfrac12 F_{(2)}^\mu\big[\nabla^2_R\phi_\mu\big] 
-
\int (\nabla_0\xi^{\mu0\iSa})
\nabla_\iSa\,
\phi_\mu
\,d\sigma
\label{DivF_lm_F2_Del1}
\\
F_{(2)}^\mu\big[\phi_\mu\big] 
&=
F_{(2)}^\mu\big[\nabla_R\phi_\mu\big]
-
\int \xi^{\rho0\iSb}\,
R^\mu{}_{\rho0\iSb}\,
\phi_\mu
\,d\sigma 
-
\tfrac12
\int \xi^{\rho a \iSb}\,
R^\mu{}_{\rho a \iSb}\,
\phi_\mu
\,d\sigma 
\label{DivF_lm_F2_Del0}
\DEfullstop
\end{align}

Now consider $F^\mu_{(3)}[\phi_\mu]$.
\begin{align}
F_{(3)}^\mu\big[\phi_\mu\big]
&= -\frac{1}{2}
\int \xi^{\mu\nu\rho\sigma}\,
\nabla_\sigma\nabla_\rho\nabla_\nu\,
\phi_\mu
\,d\sigma
= -\frac{1}{2}
\int \xi^{\mu\iSa\iSb\iSc}\,
\nabla_\iSc\nabla_\iSb\nabla_\iSa\,
\phi_\mu
\,d\sigma
- \frac{1}{2}
\int \xi^{\mu0\iSb\iSc}\,
\nabla_\iSc\nabla_\iSb\nabla_0\,
\phi_\mu
\,d\sigma
\label{DivF_lm_F3_step1}
\DEfullstop
\end{align}
The second term on the right hand side of (\ref{DivF_lm_F3_step1}) gives
\begin{align*}
&\int \xi^{\mu0\iSb\iSc}\,
\nabla_\iSc\nabla_\iSb\nabla_0\,
\phi_\mu
\,d\sigma
\\&=
\int \xi^{\mu0\iSb\iSc}\,
\nabla_\iSc\nabla_0\nabla_\iSb\,
\phi_\mu
\,d\sigma
-
\int \xi^{\mu0\iSb\iSc}\,
\nabla_\iSc (R^\rho{}_{\mu0\iSb}\,
\phi_\rho)
\,d\sigma \\
&=
\int \xi^{\mu0\iSb\iSc}\,
\nabla_0\nabla_\iSc\nabla_\iSb\,
\phi_\mu
\,d\sigma
-
\int \xi^{\mu0\iSb\iSc}\,
\Big(
R^\rho{}_{\iSb 0\iSc} \nabla_\rho\,
\phi_\mu
+
R^\rho{}_{\mu0\iSc} \nabla_\iSb\,
\phi_\rho
+
\nabla_\iSc R^\rho{}_{\mu0\iSb}\,
\phi_\rho 
+
R^\rho{}_{\mu0\iSb}\,
\nabla_\iSc \phi_\rho
\Big)\,d\sigma 
\\&=
\int \xi^{\mu0\iSb\iSc}\,
\nabla_0\nabla_\iSc\nabla_\iSb\,
\phi_\mu
\,d\sigma
-
\int \xi^{\mu0\iSb\iSc}\,
\Big(
R^{0}{}_{\iSb 0\iSc} \nabla_{0}\,
\phi_\mu 
+ R^{d}{}_{\iSb 0\iSc} \nabla_{d}\,
\phi_\mu
+
R^{\rho}{}_{\mu0\iSc} \nabla_\iSb\,
\phi_\rho 
\\&\qquad\qquad
\qquad\qquad
\qquad\qquad
\qquad\qquad
\qquad\qquad\qquad\qquad
+
\nabla_\iSc R^\rho{}_{\mu0\iSb}\,
\phi_\rho 
 +
R^\rho{}_{\mu0\iSb}\,
\nabla_\iSc \phi_\rho
\Big)\,d\sigma 
\\&=
-\int \nabla_0(\xi^{\mu0\iSb\iSc})\,
\nabla^2_{\iSc\iSb}\,
\phi_\mu
\,d\sigma
+
\int \nabla_{0} \Big( \xi^{\mu0\iSb\iSc}\,
R^{0}{}_{\iSb 0\iSc} \Big) \,
\phi_\mu \, d\sigma
\\&\qquad
-
\int \xi^{\mu0\iSb\iSc}\,
\Big( R^{d}{}_{\iSb 0\iSc} \nabla_{d}\,
\phi_\mu 
+
R^{\rho}{}_{\mu0\iSc} \nabla_\iSb\,
\phi_\rho +
\nabla_\iSc R^\rho{}_{\mu0\iSb}\,
\phi_\rho 
 +
R^\rho{}_{\mu0\iSb}\,
\nabla_\iSc \phi_\rho
\Big)\,d\sigma
\\&=
-\int \nabla_0(\xi^{\mu 0\iSb\iSc})\,
\nabla^2_{\iSc\iSb}\,
\phi_\mu
\,d\sigma
+
\int 
\Big(
\nabla_{0} \big( \xi^{\mu 0\iSb\iSc}\,
R^{0}{}_{\iSb 0\iSc} \big) \,
+
\xi^{\rho0\iSb\iSc}\,
(\nabla_\iSc R^\mu{}_{\rho0\iSb})
\Big)
\phi_\mu \, d\sigma
\\&\qquad
-
\int 
\Big( 
\xi^{\mu 0\iSb\iSc}\,
R^{a}{}_{\iSb 0\iSc} 
+
2 \xi^{\rho0\iSa\iSc}\, R^{\mu}{}_{\rho0\iSc} 
\Big)\nabla_\iSa\phi_\mu 
\,d\sigma
\DEfullstop
\end{align*}
For the first term on the right hand side of (\ref{DivF_lm_F3_step1}),
it is necessary to symmetrise \TSchange{the $\iSa\iSb\iSc$ of
$\xi^{\iSa\iSb\iSc}$}. 
\begin{align*}
\int \xi^{\mu a \iSb\iSc}\,
\nabla_\iSc\nabla_\iSb\nabla_a\,
\phi_\mu
\,d\sigma
&=
\int \xi^{\mu a \iSb\iSc}\,
\nabla_\iSc\nabla_a\nabla_\iSb\,
\phi_\mu
\,d\sigma
-
\int \xi^{\mu a \iSb\iSc}\,
\nabla_\iSc (R^\rho{}_{\mu a \iSb}\,
\phi_\rho)
\,d\sigma 
\\&=
\int \xi^{\mu a \iSb\iSc}\,
\nabla_\iSc\nabla_a\nabla_\iSb\,
\phi_\mu
\,d\sigma
-
\int \xi^{\mu a \iSb\iSc}\,
\Big(
(\nabla_\iSc R^\rho{}_{\mu a \iSb})
\phi_\rho 
 +
R^\rho{}_{\mu a \iSb}\,
\nabla_\iSc \phi_\rho
\Big)\,d\sigma 
\\&=
\int \xi^{\mu a \iSb\iSc}\,
\nabla_\iSc\nabla_a\nabla_\iSb\,
\phi_\mu
\,d\sigma
-
\int \xi^{\rho c \iSb\iSa}\,
\Big(
(\nabla_\iSa R^\mu{}_{\rho c \iSb})
\phi_\mu
 +
R^\mu{}_{\rho c \iSb}\,
\nabla_\iSa \phi_\mu
\Big)\,d\sigma 
\DEfullstop
\end{align*}
Continuing
\begin{align*}
&\int \xi^{\mu a \iSb\iSc}\,
\nabla_\iSc\nabla_\iSb\nabla_a\,
\phi_\mu
\,d\sigma
\\&=
\int \xi^{\mu a \iSb\iSc}\,
\nabla_a \nabla_\iSc\nabla_\iSb\,
\phi_\mu
\,d\sigma
-
\int \xi^{\mu a \iSb\iSc}\,
\Big(
R^\rho{}_{\iSb a \iSc} \nabla_\rho\,
\phi_\mu
+
R^\rho{}_{\mu a \iSc} \nabla_\iSb\,
\phi_\rho
+
(\nabla_\iSc R^\rho{}_{\mu a \iSb})
\phi_\rho 
 +
R^\rho{}_{\mu a \iSb}\,
\nabla_\iSc \phi_\rho
\Big)\,d\sigma 
\\&=
\int \xi^{\mu a \iSb\iSc}\,
\nabla_a \nabla_\iSc\nabla_\iSb\,
\phi_\mu
\,d\sigma
-
\int \xi^{\mu a \iSb\iSc}\,
\Big(
R^\rho{}_{\iSb a \iSc} \nabla_\rho\,
\phi_\mu
+
2R^\rho{}_{\mu a \iSc} \nabla_\iSb\,
\phi_\rho
+
(\nabla_\iSc R^\rho{}_{\mu a \iSb})
\phi_\rho 
\Big)\,d\sigma 
\\&=
\int \xi^{\mu a \iSb\iSc}\,
\nabla_a \nabla_\iSc\nabla_\iSb\,
\phi_\mu
\,d\sigma
-
\int \xi^{\mu a \iSb\iSc}\,
\Big(
R^{0}{}_{\iSb a \iSc} \nabla_{0}\,
\phi_\mu 
+ R^{d}{}_{\iSb a \iSc} \nabla_{d}\,
\phi_\mu
+
2 R^{\rho}{}_{\mu a \iSc} \nabla_\iSb\,
\phi_\rho 
 +
(\nabla_\iSc R^\rho{}_{\mu a \iSb})
\phi_\rho 
\Big)\,d\sigma 
\\&=
\int \xi^{\mu a \iSb\iSc}\,
\nabla_a \nabla_\iSc\nabla_\iSb\,
\phi_\mu
\,d\sigma
+
\int \nabla_{0} \Big( \xi^{\mu a \iSb\iSc}\,
R^{0}{}_{\iSb a \iSc} \Big) \,
\phi_\mu \, d\sigma
\\&\quad
-
\int \xi^{\mu a \iSb\iSc}\,
\Big( R^{d}{}_{\iSb a \iSc} \nabla_{d}\,
\phi_\mu 
 +
2 R^{\rho}{}_{\mu a \iSc} \nabla_\iSb\,
\phi_\rho 
+
(\nabla_\iSc R^\rho{}_{\mu a \iSb})
\phi_\rho 
\Big)\,d\sigma
\\&=
\int \xi^{\mu a \iSb\iSc}\,
\nabla_a \nabla_\iSc\nabla_\iSb\,
\phi_\mu
\,d\sigma
+
\int 
\Big(
-\xi^{\mu d \iSb\iSc}\,
R^{a}{}_{\iSb d \iSc} 
-
2 \xi^{\rho b \iSa\iSc}\,
R^{\mu}{}_{\rho b \iSc} 
\Big)
\nabla_\iSa\,\phi_\mu 
\,d\sigma
\\&\qquad
+
\int\Big(
\nabla_{0} \big( \xi^{\mu a \iSb\iSc}\,
R^{0}{}_{\iSb a \iSc} \big) \,
-
\xi^{\rho a \iSb\iSc}\,
(\nabla_\iSc R^\mu{}_{\rho a \iSb})
\Big)
\phi_\mu \,d\sigma
\DEfullstop
\end{align*}
Hence
\begin{align*}
3  \int \xi^{\mu a \iSb\iSc}\,
&\nabla_\iSc\nabla_\iSb\nabla_a\,
\phi_\mu
\,d\sigma
\\&=
\int \xi^{\mu a \iSb\iSc}\,
\nabla_\iSc\nabla_\iSb\nabla_a\,
\phi_\mu
\,d\sigma
\\&\qquad
+
\int \xi^{\mu a \iSb\iSc}\,
\nabla_\iSc\nabla_a\nabla_\iSb\,
\phi_\mu
\,d\sigma
-
\int \xi^{\rho c \iSb\iSa}\,
\Big(
(\nabla_\iSa R^\mu{}_{\rho c \iSb})
\phi_\mu
 +
R^\mu{}_{\rho c \iSb}\,
\nabla_\iSa \phi_\mu
\Big)\,d\sigma 
\\&\qquad+
\int \xi^{\mu a \iSb\iSc}\,
\nabla_a \nabla_\iSc\nabla_\iSb\,
\phi_\mu
\,d\sigma
+
\int 
\Big(
-\xi^{\mu d \iSb\iSc}\,
R^{a}{}_{\iSb d \iSc} 
-
2 \xi^{\rho b \iSa\iSc}\,
R^{\mu}{}_{\rho b \iSc} 
\Big)
\nabla_\iSa\,\phi_\mu 
\,d\sigma
\\&\qquad
+
\int\Big(
\nabla_{0} \big( \xi^{\mu a \iSb\iSc}\,
R^{0}{}_{\iSb a \iSc} \big) \,
-
\xi^{\rho a \iSb\iSc}\,
(\nabla_\iSc R^\mu{}_{\rho a \iSb})
\Big)
\phi_\mu \,d\sigma
\\&=
3 \int \xi^{\mu (a \iSb\iSc)}\,
\nabla_{(a b c)}
\phi_\mu
\,d\sigma
+
\int 
\Big(
-
\xi^{\rho c \iSb\iSa}\,
R^\mu{}_{\rho c \iSb}
-\xi^{\mu d \iSb\iSc}\,
R^{a}{}_{\iSb d \iSc} 
-
2 \xi^{\rho b \iSa\iSc}\,
R^{\mu}{}_{\rho b \iSc} 
\Big)
\nabla_\iSa\,\phi_\mu 
\,d\sigma
\\&\qquad
+
\int\Big(
-\xi^{\rho c \iSb\iSa}\, (\nabla_\iSa R^\mu{}_{\rho c \iSb})
+\nabla_{0} \big( \xi^{\mu a \iSb\iSc}\,
R^{0}{}_{\iSb a \iSc} \big) \,
-
\xi^{\rho a \iSb\iSc}\,
(\nabla_\iSc R^\mu{}_{\rho a \iSb})
\Big)
\phi_\mu \,d\sigma
\\&=
3 \int \xi^{\mu (a \iSb\iSc)}\,
\nabla_{(a b c)}
\phi_\mu
\,d\sigma
+
\int 
\Big(
-\xi^{\mu d \iSb\iSc}\,
R^{a}{}_{\iSb d \iSc} 
-
3 \xi^{\rho b \iSa\iSc}\,
R^{\mu}{}_{\rho b \iSc} 
\Big)
\nabla_\iSa\,\phi_\mu 
\,d\sigma
\\&\qquad
+
\int\Big(
\nabla_{0} \big( \xi^{\mu a \iSb\iSc}\,
R^{0}{}_{\iSb a \iSc} \big) \,
-
2\xi^{\rho a \iSb\iSc}\,
(\nabla_\iSc R^\mu{}_{\rho a \iSb})
\Big)
\phi_\mu \,d\sigma
\DEfullstop
\end{align*}
Hence (\ref{DivF_lm_F3_step1}) becomes
\begin{align*}
F_{(3)}^\mu\big[\phi_\mu\big]
&=
\tfrac12\int \nabla_0(\xi^{\mu 0\iSb\iSc})\,
\nabla^2_{\iSc\iSb}\,
\phi_\mu
\,d\sigma
-\tfrac12
\int 
\Big(
\nabla_{0} \big( \xi^{\mu 0\iSb\iSc}\,
R^{0}{}_{\iSb 0\iSc} \big) \,
+
\xi^{\rho0\iSb\iSc}\,
(\nabla_\iSc R^\mu{}_{\rho0\iSb})
\Big)
\phi_\mu \, d\sigma
\\&\qquad
+\tfrac12
\int 
\Big( 
\xi^{\mu 0\iSb\iSc}\,
R^{a}{}_{\iSb 0\iSc} 
+
2 \xi^{\rho0\iSa\iSc}\, R^{\mu}{}_{\rho0\iSc} 
\Big)\nabla_\iSa\phi_\mu 
\,d\sigma
\\&\qquad
-\tfrac12 \int \xi^{\mu (a \iSb\iSc)}\,
\nabla^3_{a b c}
\phi_\mu
\,d\sigma
+\tfrac16 
\int 
\Big(
\xi^{\mu d \iSb\iSc}\,
R^{a}{}_{\iSb d \iSc} 
+
3 \xi^{\rho b \iSa\iSc}\,
R^{\mu}{}_{\rho b \iSc} 
\Big)
\nabla_\iSa\,\phi_\mu 
\,d\sigma
\\&\qquad
+\tfrac16
\int\Big(
-\nabla_{0} \big( \xi^{\mu a \iSb\iSc}\,
R^{0}{}_{\iSb a \iSc} \big) \,
+
2\xi^{\rho a \iSb\iSc}\,
(\nabla_\iSc R^\mu{}_{\rho a \iSb})
\Big)
\phi_\mu \,d\sigma
\\&=
-\tfrac12 \int \xi^{\mu (a \iSb\iSc)}\,
\nabla^3_{a b c}
\phi_\mu
\,d\sigma
+
\tfrac12\int \nabla_0(\xi^{\mu 0\iSb\iSc})\,
\nabla^2_{\iSc\iSb}\,
\phi_\mu
\,d\sigma
\\&\qquad
+
\int 
\Big( 
\tfrac12
\xi^{\mu 0\iSb\iSc}\,
R^{a}{}_{\iSb 0\iSc} 
+
 \xi^{\rho0\iSa\iSc}\, R^{\mu}{}_{\rho0\iSc} 
R^\mu{}_{\rho c \iSb}
+\tfrac16\xi^{\mu d \iSb\iSc}\,
R^{a}{}_{\iSb d \iSc} 
+
\tfrac12 \xi^{\rho b \iSa\iSc}\,
R^{\mu}{}_{\rho b \iSc} 
\Big)
\nabla_\iSa\,\phi_\mu 
\,d\sigma
\\&\qquad
+
\int 
\Big(
-\tfrac12\nabla_{0} \big( \xi^{\mu 0\iSb\iSc}\,
R^{0}{}_{\iSb 0\iSc} \big) \,
-\tfrac12
\xi^{\rho0\iSb\iSc}\,
(\nabla_\iSc R^\mu{}_{\rho0\iSb})
-\tfrac16\nabla_{0} \big( \xi^{\mu a \iSb\iSc}\,
R^{0}{}_{\iSb a \iSc} \big) \,
+\tfrac13
\xi^{\rho a \iSb\iSc}\,
(\nabla_\iSc R^\mu{}_{\rho a \iSb})
\Big)
\phi_\mu \,d\sigma
\DEfullstop
\end{align*}
Splitting $F^\mu_{(3)}[\phi_\mu]$ gives
\begin{align}
F^\mu_{(3)}[\nabla^3_R\phi_\mu]
&=
-3 \int \xi^{\mu (a \iSb\iSc)}\,
\nabla^3_{a b c}
\phi_\mu
\,d\sigma
\label{DivF_lm_F3_Del3}
\DEcomma
\\
F^\mu_{(3)}[\nabla^2_R\phi_\mu]
&=
\tfrac13 F^\mu_{(3)}[\nabla^3_R \phi_\mu]
+
\int \nabla_0(\xi^{\mu 0\iSb\iSc})\,
\nabla^2_{\iSc\iSb}\,
\phi_\mu
\,d\sigma
\label{DivF_lm_F3_Del2}
\DEcomma
\\
F^\mu_{(3)}[\nabla_R\phi_\mu]
&=
\tfrac12 F^\mu_{(3)}[\nabla^2_R \phi_\mu]
+
\int 
\Big( 
\tfrac12
\xi^{\mu 0\iSb\iSc}\,
R^{a}{}_{\iSb 0\iSc} 
+
 \xi^{\rho0\iSa\iSc}\, R^{\mu}{}_{\rho0\iSc} 
+\tfrac16\xi^{\mu d \iSb\iSc}\,
R^{a}{}_{\iSb d \iSc} 
+
\tfrac12 \xi^{\rho b \iSa\iSc}\,
R^{\mu}{}_{\rho b \iSc} 
\Big)
\nabla_\iSa\,\phi_\mu 
\,d\sigma 
\label{DivF_lm_F3_Del1}
\DEcomma
\\
F^\mu_{(3)}[\phi_\mu]
&=
F^\mu_{(3)}[\nabla_R \phi_\mu]
+
\int 
\Big(
-\tfrac12\nabla_{0} \big( \xi^{\mu 0\iSb\iSc}\,
R^{0}{}_{\iSb 0\iSc} \big) \,
-\tfrac12
\xi^{\rho0\iSb\iSc}\,
(\nabla_\iSc R^\mu{}_{\rho0\iSb})
-\tfrac16\nabla_{0} \big( \xi^{\mu a \iSb\iSc}\,
R^{0}{}_{\iSb a \iSc} \big) \,
\notag
\\&\qquad\qquad\qquad\qquad
+\tfrac13
\xi^{\rho a \iSb\iSc}\,
(\nabla_\iSc R^\mu{}_{\rho a \iSb})
\Big)
\phi_\mu \,d\sigma
\label{DivF_lm_F3_Del0}
\DEfullstop
\end{align}


Since $F^\mu=0$, then from corollary \ref{lm_J_0}, 
\begin{align}
0&=F^\mu[\nabla^3_R\phi_\mu]
=F^\mu_{(3)}[\nabla^3_R\phi_\mu]
\DEcomma
\label{DivF_lm_F3}
\\
0&=
F^\mu[\nabla^2_R\phi_\mu]
=F^\mu_{(3)}[\nabla^2_R\phi_\mu] 
+ F^\mu_{(2)}[\nabla^2_R\phi_\mu]
\DEcomma
\label{DivF_lm_F2}
\\
0&=F^\mu[\nabla_R\phi_\mu]=
F^\mu_{(3)}[\nabla_R\phi_\mu]+
F^\mu_{(2)}[\nabla_R\phi_\mu]+
F^\mu_{(1)}[\nabla_R\phi_\mu]
\DEcomma
\label{DivF_lm_F1}
\\
0&=F^\mu[\phi_\mu]=
F^\mu_{(3)}[\phi_\mu]+
F^\mu_{(2)}[\phi_\mu]+
F^\mu_{(1)}[\phi_\mu]
\label{DivF_lm_F0}
\DEfullstop
\end{align}
From (\ref{DivF_lm_F3}) and (\ref{DivF_lm_F3_Del3}) we get
(\ref{DivF_xi_constraint_adapt}).

From (\ref{DivF_lm_F2}), (\ref{DivF_lm_F2_Del2}) and
(\ref{DivF_lm_F3_Del2}) we can calculate $\nabla_0\xi^{\mu 0 \iSa\iSb}$
\begin{align*}
0
=F^\mu_{(3)}[\nabla^2_R\phi_\mu] 
+ F^\mu_{(2)}[\nabla^2_R\phi_\mu]
&=
2 \int \xi^{\mu (a \iSb)}\,
\nabla_{ab}
\phi_\mu
\,d\sigma
+
\int \nabla_0(\xi^{\mu0\iSa\iSb})\,
\nabla^2_{\iSa\iSb}\,
\phi_\mu
\,d\sigma
\DEcomma
\end{align*}
which gives (\ref{DivF_xi4dot_adapt}).

From (\ref{DivF_lm_F1}), (\ref{DivF_lm_F1_Del1}),
(\ref{DivF_lm_F2_Del1}) and (\ref{DivF_lm_F3_Del1}) we can calculate
$\nabla_0\xi^{\mu 0 \iSa}$ 
\begin{align*}
0&=F^\mu[\nabla_R\phi_\mu]=
F^\mu_{(3)}[\nabla_R\phi_\mu]+
F^\mu_{(2)}[\nabla_R\phi_\mu]+
F^\mu_{(1)}[\nabla_R\phi_\mu]
\\&=
-
\int \xi^{\mu a}\,
\nabla_a\phi_\mu \,d\sigma 
+
\tfrac12 F_{(2)}^\mu\big[\nabla^2_R\phi_\mu\big] 
-
\int (\nabla_0\xi^{\mu0\iSa})
\nabla_\iSa\,
\phi_\mu
\,d\sigma
\\&\quad+
\tfrac12 F^\mu_{(3)}[\nabla^2_R \phi_\mu]
+
\int 
\Big( 
\tfrac12
\xi^{\mu 0\iSb\iSc}\,
R^{a}{}_{\iSb 0\iSc} 
+
 \xi^{\rho0\iSa\iSc}\, R^{\mu}{}_{\rho0\iSc} 
+
+\tfrac16\xi^{\mu d \iSb\iSc}\,
R^{a}{}_{\iSb d \iSc} 
\notag
+
\tfrac12 \xi^{\rho b \iSa\iSc}\,
R^{\mu}{}_{\rho b \iSc} 
\Big)
\nabla_\iSa\,\phi_\mu 
\,d\sigma 
\\&=
\int \Big(
-
\xi^{\mu a}\,
-
(\nabla_0\xi^{\mu0\iSa})
+
\tfrac12
\xi^{\mu 0\iSb\iSc}\,
R^{a}{}_{\iSb 0\iSc} 
+
 \xi^{\rho0\iSa\iSc}\, R^{\mu}{}_{\rho0\iSc} 
+\tfrac16\xi^{\mu d \iSb\iSc}\,
R^{a}{}_{\iSb d \iSc} 
+
\tfrac12 \xi^{\rho b \iSa\iSc}\,
R^{\mu}{}_{\rho b \iSc} 
\Big)
\nabla_\iSa\,\phi_\mu 
\,d\sigma 
\DEcomma
\end{align*}
which gives (\ref{DivF_xi3dot_adapt}).

From (\ref{DivF_lm_F0}), (\ref{DivF_lm_F1_Del0}),
(\ref{DivF_lm_F2_Del0}) and (\ref{DivF_lm_F3_Del0}) we have
\begin{align*}
0&=F^\mu[\phi_\mu]=
F^\mu_{(3)}[\phi_\mu]+
F^\mu_{(2)}[\phi_\mu]+
F^\mu_{(1)}[\phi_\mu]
\\&=
F^\mu_{(1)}[\nabla_R \phi_\mu]
+
 \int (\nabla_0\xi^{\mu 0})\,
\phi_\mu
\,d\sigma
\\&\quad+
F_{(2)}^\mu\big[\nabla_R\phi_\mu\big]
-
\int \xi^{\rho0\iSb}\,
R^\mu{}_{\rho0\iSb}\,
\phi_\mu
\,d\sigma 
-
\tfrac12
\int \xi^{\rho a \iSb}\,
R^\mu{}_{\rho a \iSb}\,
\phi_\mu
\,d\sigma 
\\&\quad+
F^\mu_{(3)}[\nabla_R \phi_\mu]
+
\int 
\Big(
-\tfrac12\nabla_{0} \big( \xi^{\mu 0\iSb\iSc}\,
R^{0}{}_{\iSb 0\iSc} \big) \,
-\tfrac12
\xi^{\rho0\iSb\iSc}\,
(\nabla_\iSc R^\mu{}_{\rho0\iSb})
-\tfrac16\nabla_{0} \big( \xi^{\mu a \iSb\iSc}\,
R^{0}{}_{\iSb a \iSc} \big) \,
\notag
\\&\qquad\qquad\qquad\qquad
+\tfrac13
\xi^{\rho a \iSb\iSc}\,
(\nabla_\iSc R^\mu{}_{\rho a \iSb})
\Big)
\phi_\mu \,d\sigma
\\&=
\int \Big(
(\nabla_0\xi^{\mu 0})\,
-
\xi^{\rho0\iSb}\,
R^\mu{}_{\rho0\iSb}\,
-
\tfrac12
\xi^{\rho a \iSb}\,
R^\mu{}_{\rho a \iSb}\,
-\tfrac12\nabla_{0} \big( \xi^{\mu 0\iSb\iSc}\,
R^{0}{}_{\iSb 0\iSc} \big) \,
-\tfrac12
\xi^{\rho0\iSb\iSc}\,
(\nabla_\iSc R^\mu{}_{\rho0\iSb})
-\tfrac16\nabla_{0} \big( \xi^{\mu a \iSb\iSc}\,
R^{0}{}_{\iSb a \iSc} \big) \,
\notag
\\&\qquad\qquad\qquad\qquad
+\tfrac13
\xi^{\rho a \iSb\iSc}\,
(\nabla_\iSc R^\mu{}_{\rho a \iSb})
\Big)
\phi_\mu \,d\sigma
\DEcomma
\end{align*}
which gives
(\ref{DivF_xi2dot_adapt})

\end{proof}

As stated the key advantage with using the Dixon representation is
that the components are tensors. Therefore writing
(\ref{DivF_xi4dot_adapt})-(\ref{DivF_xi_constraint_adapt}) using
spacetime indices, results in tensor equations which are valid for all
coordinate systems. We introduce the spatial projection tensor
\begin{align}
\Proj^{\rho}_{\alpha} = \delta^{\rho}_{\alpha} - \Cdot^{\rho}\,\DixVec_{\alpha}
\label{DivF_def_Proj}
\DEfullstop
\end{align}

\begin{theorem}
\label{thm_DixQuad_Ten}
The \TSchange{divergenceless} condition (\ref{SE_Div_Free}) \TSchange{corresponds} to the
following tensor equations for the components
\begin{align}
\nabla_\Cdot (\DixVec_\nu \,\xi^{\mu\nu\rho\sigma})
&=
-2 \,\Proj^{\rho}_{\beta}\, \Proj^{\sigma}_{\alpha} \,\xi^{\mu(\beta\alpha)}
\label{DivF_xi4dot}
\DEcomma
\\
\nabla_\Cdot (\DixVec_\nu \,\xi^{\mu\nu\rho})
&=
\Proj^{\rho}_{\alpha} 
\big(-\xi^{\mu \alpha} 
+
\tfrac12 \DixVec_\nu \Cdot^{\beta} \,
  \xi^{\mu \nu \lambda \sigma} R^{\alpha}{}_{\lambda \beta \sigma}
+
(\DixVec_\nu \Cdot^\beta + \tfrac12\Proj^\beta_\nu)\, 
\xi^{\lambda \nu \sigma \alpha} 
R^{\mu}{}_{\lambda\beta\sigma} 
+
\tfrac16\Proj^\beta_\nu\, 
\xi^{\alpha \nu \lambda \sigma} R^{\mu}{}_{\lambda \beta \sigma}
\big)
\DEcomma
\label{DivF_xi3dot}
\\
\nabla_\Cdot (\DixVec_\nu \,\xi^{\mu\nu})
&=
(\DixVec_\nu\Cdot^\beta +\tfrac12\Proj^\beta_\nu)
\xi^{\rho \nu \lambda}R^{\mu}{}_{\rho \beta \lambda}
+\nabla_\Cdot \big(
(\tfrac12\DixVec_\nu\Cdot^\beta +\tfrac16\Proj^\beta_\nu)
\DixVec_{\zeta}\, 
\xi^{\mu \nu \lambda\sigma}\,R^{\zeta}{}_{\lambda \beta \sigma} 
\big)
\notag
\\&\qquad\qquad
+(\tfrac12\DixVec_\nu\Cdot^\beta -\tfrac13\Proj^\beta_\nu)\,
\xi^{\rho \nu \lambda\sigma}\,
(\nabla_\lambda R^\mu{}_{\rho \beta \sigma})
\label{DivF_xi2dot}
\DEcomma
\end{align}
together with the constraint
\begin{align}
\Proj^\beta_\nu\,\Proj^\alpha_\rho\,
\Proj^\lambda_\sigma\,\xi^{\mu(\nu\rho\sigma)}
= 0
\label{DivF_xi_constraint}
\DEfullstop
\end{align}
\end{theorem}
\begin{proof}
The equations
(\ref{DivF_xi4dot_adapt})-(\ref{DivF_xi_constraint_adapt}) are
replaced by 
(\ref{DivF_xi4dot})-(\ref{DivF_xi_constraint}) as follows. 
$\nabla_0$ is replaced by $\nabla_\Cdot$.
Each lower index $0$ it is necessary to contract with
$\Cdot^\beta$ and each upper index $0$ it is necessary to contract with
$\DixVec_\nu$. Each spatial index $a,b,\ldots$ it is necessary to
project out using $\Proj^\beta_\nu$. However if the spatial index is
one the third or fourth index of $\xi^{\rho \nu a}$ or  
$\xi^{\rho \nu a b}$ the projection is not necessary.

As a result, in the adapted coordinate system where
$\Cdot^\beta=\delta^\beta_0$ and $\DixVec_\nu=\delta^0_\nu$ then
(\ref{DivF_xi4dot})-(\ref{DivF_xi_constraint})
become
(\ref{DivF_xi4dot_adapt})-(\ref{DivF_xi_constraint_adapt}).
However, (\ref{DivF_xi4dot})-(\ref{DivF_xi_constraint}) are clearly
tensorial equations and therefore true for all coordinate systems.
\end{proof}

As was observed in \cite{gratus2020distributional}, the equations for the components arising
from the \TSchange{divergenceless} condition of the stress-energy tensor are
insufficient to completely determine the dynamics of the components.
To see this, from
(\ref{SE_Xi_Symm_quad}) we have  (10+40+100=150) components. From
(\ref{Basic_Dix_Cons_quaf}) this reduces to (150-10-40=100)
components. Using either (\ref{DivF_xi_constraint_adapt}) or
(\ref{DivF_xi_constraint}) gives us 40 constraints, hence there are 60
components. Equations
(\ref{DivF_xi4dot_adapt})-(\ref{DivF_xi2dot_adapt}) or
(\ref{DivF_xi4dot})-(\ref{DivF_xi2dot}) give 40 ODEs. Hence there are
20 free components.
These still have to be determined via
constitutive relations from a model of the underlying material.

These are similar to the equations have been found by Steinhoff and
Puetzfeld \cite{steinhoff2010multipolar}. However their equations are
implicit. On the right  (\ref{DivF_xi2dot}) we see that there is a
covariant derivative of $\xi^{\rho \nu \lambda\sigma}$. However this
term can be expanded out to identified $\nabla_\Cdot\xi^{\rho \nu
  \lambda\sigma}$. One can then substitute in (\ref{DivF_xi4dot}) and
the appropriate constitutive relations.

\begin{corrol}

As a simple consistency check we see if we reduce to a dipole. Set 
\begin{align}
\xi^{\rho \nu \lambda\sigma} = 0
,\quad
\xi^{\rho \nu \lambda} 
= X^{\lambda}\dot{C}^{\rho}\dot{C}^{\nu}+S^{\lambda ( \rho}\dot{C}^{\nu )}
\quadand
\xi^{\rho \nu} = 
2 P^{( \rho}\dot{C}^{\nu )} -2 m \dot{C}^{\rho} \dot{C}^{\nu} 
\label{DivF_dipole_xi}
\DEcomma
\end{align}
where $\DixVec_\nu=-\Cdot_\nu$,
$m$ is the rest mass, 
$X^{\iMa}$ is the displacement vector,
$P^{\iMa}$ is the  rate of change of the displacement vector and 
$S^{\iMa\iMb}$ is the spin tensor satisfy
\begin{align}
X_{\iMa}\,\Cdot^{\iMa}=0
,\quad
P_{\iMa}\,\Cdot^{\iMa}=0
,\quad
\Cdot_{\iMa}\,S^{\iMa\iMb}=0
\quadand
S^{\iMa\iMb}+S^{\iMb\iMa}=0
\label{DivF_diploe_constr}
\DEfullstop
\end{align}
Then we get the Mathisson-Papapetrou-Tulczyjew-Dixon for a dipole
along a geodesics.
\begin{align}
\dot{m}=0
\,,\quad
\nabla_\Cdot{X^{\iMa}}
= -P^{\iMa}
\,,\quad
\nabla_\Cdot{P^{\iMa}} 
= 
\tfrac12 R^{\iMa}{}_{\iMb\iMc\iMd}\,\Cdot^\iMb\, S^{\iMd\iMc} 
+R^{\iMa}{}_{\iMb\iMc\iMd}\,\Cdot^\iMb\, \Cdot^\iMc\,X^\iMd
\quadand
\nabla_\Cdot{S^{\iMa\iMb}}
&=
0
\label{DivF_diploe_DSab}
\DEcomma
\end{align}
\end{corrol}

\begin{proof}
Since $C$ is a geodesic and $\DixVec_\nu=-\Cdot_\nu$ then $\nabla_\Cdot
\DixVec_\nu=0$. From (\ref{DivF_dipole_xi}) we have
\begin{align*}
\nabla_\Cdot (\DixVec_\nu \,\xi^{\mu\nu\rho})
=
\nabla_\Cdot
\Big(\DixVec_\nu \big(X^{\rho}\dot{C}^{\mu}\dot{C}^{\nu}+S^{\rho ( \mu}\dot{C}^{\nu )}\big)\Big)
=
\nabla_\Cdot
\big(X^{\rho}\dot{C}^{\mu}
+\tfrac12 S^{\rho  \mu}
\big)
=
\dot{C}^{\mu} \nabla_\Cdot X^{\rho} 
+
\tfrac12 \nabla_\Cdot S^{\rho  \mu}
\end{align*}
While from (\ref{DivF_xi3dot}) we have
\begin{align*}
\nabla_\Cdot (\DixVec_\nu \,\xi^{\mu\nu\rho})
&=
-\Proj^{\rho}_{\alpha} \xi^{\mu \alpha} 
=
2\Proj^{\rho}_{\alpha} 
\big( m \dot{C}^{\mu} \dot{C}^{\alpha}- P^{( \mu}\dot{C}^{\alpha )}\big)
=
-P^{\rho}\dot{C}^{\mu}
\end{align*}
giving
\begin{align*}
\dot{C}^{\mu} \nabla_\Cdot X^{\rho} 
+
\tfrac12 \nabla_\Cdot S^{\rho  \mu}
=
-P^{\rho}\dot{C}^{\mu}
\end{align*}
Projecting this out with respect to $\Cdot_\mu$ and $\pi_\mu^\alpha$
gives (\ref{DivF_diploe_DSab}.2)
and  (\ref{DivF_diploe_DSab}.4).

From (\ref{DivF_dipole_xi}.3) we have
\begin{align*}
 \n_{\dot{C}}(N_{\nu}\xi^{\mu \nu})
=
\n_{\dot{C}}\Big(N_{\nu} 
\big(P^{( \mu}\dot{C}^{\nu )}  -2 m \dot{C}^{\mu} \dot{C}^{\nu}\big)\Big)
=
\n_{\dot{C}}P^{\mu}  
-
2 \dot{C}^{\mu} \n_{\dot{C}} m 
=
\n_{\dot{C}}P^{\mu}  
-
2 \dot{C}^{\mu} \dot{m}
\n_{\dot{C}}P^{\mu} 
\end{align*}
Substituting (\ref{DivF_dipole_xi}.1) into (\ref{DivF_xi2dot}) gives
\begin{align*}
 \n_{\dot{C}}(N_{\nu}\xi^{\mu \nu})
&=
 (\DixVec_\nu\Cdot^\beta +\tfrac12\Proj^\beta_\nu)
\xi^{\rho \nu \lambda}R^{\mu}{}_{\rho \beta \lambda} 
=
 (\DixVec_\nu\Cdot^\beta +\tfrac12\Proj^\beta_\nu)
\LF X^{\lambda}\dot{C}^{\rho}\dot{C}^{\nu}
+S^{\lambda ( \rho}\dot{C}^{\nu )} \RF R^{\mu}{}_{\rho \beta \lambda}  
\\&=
 \dot{C}^{\beta}X^{\lambda}\dot{C}^{\rho}R^{\mu}{}_{\rho \beta \lambda} 
+ \frac{1}{2}N_{\nu}\dot{C}^{\beta}\dot{C}^{\nu}S^{\lambda \rho}R^{\mu}{}_{\rho \beta \lambda}
+\frac{1}{2}N_{\nu}\dot{C}^{\beta}\dot{C}^{\rho}S^{\lambda \nu}R^{\mu}{}_{\rho \beta \lambda} 
\\&\qquad\qquad
+\tfrac12\Proj^\beta_\nu \LF X^{\lambda}\dot{C}^{\rho}\dot{C}^{\nu}
+S^{\lambda ( \rho}\dot{C}^{\nu )} \RF R^{\mu}{}_{\rho \beta \lambda}  
\\&=
 X^{\lambda} \dot{C}^{\beta}\dot{C}^{\rho}R^{\mu}{}_{\rho \beta \lambda} 
+ \frac{1}{2}\dot{C}^{\beta}S^{\lambda \rho}R^{\mu}{}_{\rho \beta \lambda} 
+\tfrac12\Proj^\beta_\nu \LF X^{\lambda}\dot{C}^{\rho}\dot{C}^{\nu}
+S^{\lambda ( \rho}\dot{C}^{\nu )} \RF R^{\mu}{}_{\rho \beta \lambda}  
\\&=
 X^{\lambda} \dot{C}^{\beta}\dot{C}^{\rho}R^{\mu}{}_{\rho \beta \lambda} 
+ \frac{1}{2}\dot{C}^{\beta}S^{\lambda \rho}R^{\mu}{}_{\rho \beta \lambda} 
+ \frac{1}{4}\dot{C}^{\rho}S^{\lambda \beta}R^{\mu}{}_{\rho \beta \lambda}  
\\&=
 X^{\lambda} \dot{C}^{\beta}\dot{C}^{\rho}R^{\mu}{}_{\rho \beta \lambda} 
+ \frac{1}{4}\dot{C}^{\rho}S^{\lambda \beta}R^{\mu}{}_{\rho \beta \lambda} 
+ \frac{1}{4}\dot{C}^{\rho}S^{\lambda \beta}R^{\mu}{}_{\rho \beta \lambda}  
\\&=
 X^{\lambda} \dot{C}^{\beta}\dot{C}^{\rho}R^{\mu}{}_{\rho \beta \lambda} 
+ \frac{1}{2}\dot{C}^{\rho}S^{\lambda \beta}R^{\mu}{}_{\rho \beta \lambda}
\DEfullstop
\end{align*}
hence
\begin{align*}
\n_{\dot{C}}P^{\mu}  
-
2 \dot{C}^{\mu} \dot{m}
=
 X^{\lambda} \dot{C}^{\beta}\dot{C}^{\rho}R^{\mu}{}_{\rho \beta \lambda} 
+ \frac{1}{2}\dot{C}^{\rho}S^{\lambda \beta}R^{\mu}{}_{\rho \beta \lambda}
\end{align*}
Projecting this out with respect to $\Cdot_\mu$ and $\pi_\mu^\alpha$
gives (\ref{DivF_diploe_DSab}.1)
and  (\ref{DivF_diploe_DSab}.3).
\end{proof}



\section{Comparison with Dixon's results}
\label{ch_Dix}

\def\DelM#1{{\stackrel{M}\nabla_{\Cdot}}#1}

In his work, Dixon makes two conjectures for the dynamics of the
components of a quadrupole. Neither of these equations correspond to
the \TSchange{divergenceless} condition. Thus they are not the generalisation of
the Mathisson-Papapetrou–Tulczyjew–Dixon equations for the
quadrupole. 

Recall from \TSchange{\cite{gratus2020distributional}}, that we can identify
$\xi^{\mu\nu}=I^{\mu\nu}$,
$\xi^{\mu\nu\kappa}=-I^{\kappa\mu\nu}$ 
and $\xi^{\mu\nu\kappa\lambda}=I^{\kappa\lambda\mu\nu}$. Using
\cite[eqn. (1.37)]{DixonIII} and \cite[eqn. (2.4)]{dixon1967description} we 
can\footnote{In \cite{gratus2020distributional}, we identified
  $J^{\kappa\lambda\mu\nu}$ incorrectly.} write
$J^{\kappa\lambda\mu\nu}=I^{[\kappa[\lambda\mu]\nu]}
=
\tfrac14(
I^{\kappa\lambda\mu\nu}
-I^{\mu\lambda\kappa\nu}
-I^{\kappa\nu\mu\lambda}
+I^{\mu\nu\kappa\lambda})$.

In \cite[eqns. (7.34)-(7.37)]{dixon1970dynamics}  Dixon proposes a 
simple rotational dynamics. Here he introduces a connection (7.18), $\DelM$
which he writes as $\stackrel{(m)}{\frac{\delta}{ds}}$
\begin{align}
\DelM{B^\kappa} 
=
\nabla_\Cdot B^\kappa - 
(\dot{u}^\kappa\,u^\lambda -\dot{u}^\lambda\,u^\kappa )
B^\lambda
\label{Dix_def_nabla_M}
\DEcomma
\end{align}
where $u^\kappa$ is described as the body's dynamical velocity.
Using this connection we can define a rotation tensor
$\chi\Omega_{\lambda\kappa}$ where (7.34)
\begin{align}
\DelM{B^\kappa} 
=
\chi\Omega^\kappa{}_{\lambda}\, B^\lambda
,\quad
\Omega_{(\kappa\lambda)}=0
\quadand
\Omega^\kappa{}_{\lambda}\, u^\lambda = 0
\label{Dix_def_Omega}
\DEfullstop
\end{align}
From this the dynamical equation for a non-rotating quadrupole is
given by (7.36)
\begin{align}
\DelM{J^{\kappa\lambda\mu\nu}}
=
-\chi\Omega^\kappa{}_{\rho} J^{\lambda\rho\mu\nu}
+\chi\Omega^\lambda{}_{\rho} J^{\kappa\rho\mu\nu}
-\chi\Omega^\nu{}_{\rho} J^{\kappa\lambda\rho\mu}
+\chi\Omega^\mu{}_{\rho} J^{\kappa\lambda\rho\nu}
\label{Dix_J_Dyn}
\DEfullstop
\end{align}
It is clear that such a non-rotating quadrupole would not satisfy
(\ref{DivF_xi4dot})--(\ref{DivF_xi_constraint}) and therefore not
correspond to a \TSchange{divergenceless} stress-energy tensor.

\vspace{1em}

By contrast in \cite[eqn. (4.11)]{DixonII} Dixon proposes a non-dynamical equation based on symmetry. In our language this becomes
\begin{align}
\xi^{\mu(\nu\rho\sigma)}
= 0
\label{Dix_xi_constraint}
\DEfullstop
\end{align}
This is very similar to the constraint (\ref{DivF_xi_constraint}) but
without the projections. Interestingly such a constraint gives the
same number of free components as the dynamical equations arising from
the divergenceless condition.
Again such conditions on $J^{\mu\nu\rho\sigma}$ do not correspond to
the divergenceless condition.


\section{Discussion and Conclusion}
\label{ch_Concl}

The key result of this article is derivation of the dynamics of the
Dixon quadrupole, as given in section \ref{ch_DivF}.  
It is clear looking at the method used to derive the
equations, that this could be extended to arbitrary order
multipoles, giving rise to higher covariant derivatives of the
curvature. In section \ref{ch_Dix} we compared these equations with
two which had be conjectured by Dixon. We observe that since Dixon's
equations do not couple to curvature \TSchange{as they} were not compatible with
the \TSchange{divergenceless} property of the stress-energy tensor.

In order to get to derive the dynamical equations it was necessary to
establish many properties of general Dixon multipoles, 
presented in section \ref{ch_J}. These
include the fact that all multipoles can be represented as Dixon
multipoles, that the components were unique and that they can be
extracted using \TSchange{particular} test tensors. We concluded this
section showing the link between the moments of \TSchange{regular tensors and
multipoles}.

\begin{figure}
\begin{subfigure}{0.4\textwidth}
\begin{tikzpicture}
\draw [ultra thick,black,->-=0.8] (0.5,-2) to[out=80,in=-80] 
(0,0) to[out=100,in=-100] (-0.5,2) ;
\shadedraw (0,0) [rotate={-10},yscale=0.3] circle (1.5) ;
\draw [ultra thick,black,dashed] (0.5,-2) to[out=80,in=-80] (0,0) ;
\draw [ultra thick,black] (0,0) to[out=100,in=-100] (-0.5,2) ;
\draw [ultra thick,black!40!green,->] (0,0) -- node
      [pos=1.1]{$\,\,\DixVec_{\iMa}$} (100:1.5) ;
\end{tikzpicture}
\caption{The Dixon vector given by the tangent to the worldline.}
\label{fig_DV_tangent}
\end{subfigure}
\qquad
\begin{subfigure}{0.5\textwidth}
\begin{tikzpicture}
\draw [very thick,shift={(0,0)}] (-135:1.7) -- ++(45:4)
      node[above]{$q$} -- +(-45:4) ;
\draw [very thick,dashed,shift={(0,-1.25)}] (1.63,0) [yscale=0.2] circle (2.8) ;
\draw [ultra thick,black,->-=0.8] (0.5,-2) to[out=80,in=-80] 
(0,0) to[out=100,in=-100] (-0.5,2) ;
\shadedraw (0,0) [rotate={45},yscale=0.3] circle (1.5) ;
\draw [ultra thick,black,dashed] (0.5,-2) to[out=80,in=-80] (0,0) ;
\draw [ultra thick,black] (0,0) to[out=100,in=-100] (-0.5,2) ;
\draw [ultra thick,black!40!green,->] (0,0) -- node
      [above={4pt},pos=0.8]{$\,\,\DixVec_{\iMa}$} (45:1.7) ;
\end{tikzpicture}
\caption{The Dixon vector given by the backward lightcone of a distant
  observer event $q$.}
\label{fig_DV_back_light}
\end{subfigure}

\caption{The Dixon vector $\DixVec_{\iMa}$ and the corresponding
  Dixon \TSchange{geodesic hypersurfaces}, for two different scenarios.}
\label{fig_DV}
\end{figure}
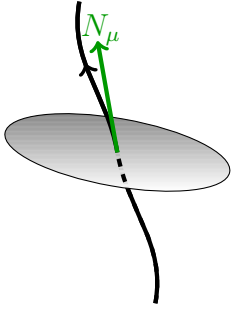
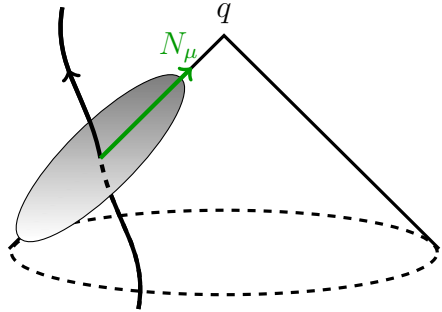

As noted throughout this article, all the results depend on the choice
of the Dixon vector.  There are multiple choices of such Dixon
vectors, and there may not be a natural one. In figure \ref{fig_DV} we
see two such choice. In figure \ref{fig_DV_tangent} we show the Dixon
vector given by the tangent of the worldline. This is a natural
choice if one is interested in the multipole dynamics as observed by
the particle itself. 

Alternatively, one may wish to model the dynamics of the
multipole as we, as distant observers, see it. As seen in figure
\ref{fig_DV_back_light}, this involves constructing the backward
lightcone from each event $q$ in our worldline. This give rise to a
lightlike Dixon vector which points from the multipole worldline, $C$
to $q$. The corresponding Dixon geodesic hypersurface is then tangent
to backward lightcone. It does not, however, coincide with the
backward lightcone. For a very distant object the two would be very
close and this may be a sufficiently good approximation. However if the
object and the observer were closer this discrepancy may become
important. It would then be necessary to use the Ellis representation
of multipoles and a coordinate system adapted to backward lightcones.
As was observed, the Dixon vector $\DixVec_{\iMa}$ in figure
\ref{fig_DV_back_light} is lightlike. This does not affect any of the
calculation as the only constraint on the Dixon vector is 
(\ref{Intro_Tab_Dixon_N_NonOrth}).

As noted in section \ref{ch_DivF}, there are 20 free components of
quadrupole. \TSchange{The equations for these components, also known as constitutive relations, will need to use additional information}. This could
be knowledge of the underlying matter which makes up the extended
object, or they could arise from observation. Now that the
tensorial expression of the dynamical equations for the moments are
known, it will help in
establishing the constitutive relations for different objects.

Once the additional equations have been chosen, one could ask how
several quadrupoles gravitationally interact. As stated in
\cite{gratus2020distributional} this interaction can only be
perturbative, where the quadrupoles interact via gravitational waves. There
is still the issue of gravitational radiation reaction. However, this
can be avoided if we only consider the action of one quadrupoles by
the fields generated by all the other quadrupoles. This is similar to
the electromagnetic case considered in \cite{gratus2022maxwell}.

\section*{Acknowledgements}

JG is grateful for the support provided by STFC (the Cockcroft
Institute ST/P002056/1 and ST/V001612/1). ST would like to thank the Faculty of
Science and Technology, Lancaster University for their support. 

   We would like to thank Alex Warwick, Lancaster University, for
   reading the manuscript and making 
suggestions.


\bibliographystyle{unsrt}
\bibliography{bibliography}

\end{document}